\documentclass{amsart}
\usepackage[margin=1in]{geometry}

\usepackage{amssymb}
\usepackage{kpfonts}
\usepackage[T1]{fontenc}
\usepackage{multicol}
\usepackage{mathrsfs}
\usepackage[usenames,dvipsnames]{xcolor}
\usepackage{multicol}
\usepackage{amsthm}
\usepackage{float}
\usepackage{tikz}
\usepackage{tikz-cd}
\usetikzlibrary{arrows}
\usetikzlibrary{shapes.arrows}   
\usetikzlibrary{positioning}
\usetikzlibrary{decorations.pathreplacing,calligraphy,patterns}

\usepackage{subcaption}
\usepackage{chessfss}

\definecolor{sangria}{rgb}{0.57, 0.0, 0.04}
\definecolor{b-violet}{RGB}{215,0,113}
\definecolor{b-blue}{RGB}{0,53,170}
\definecolor{b-mauve}{RGB}{156,78,151}

\usepackage[backend=biber,hyperref,
style=alphabetic,
giveninits = true,
block=space]{biblatex}
\renewbibmacro{in:}{} 
\usepackage{hyperref}
\hypersetup{breaklinks=true,
	colorlinks=true,
	linkcolor=b-violet, 
	urlcolor=b-blue,
	citecolor=b-mauve,
}

\newtheorem{theorem}{Theorem}
\newtheorem{corollary}[theorem]{Corollary}
\newtheorem{lemma}[theorem]{Lemma}

\newtheorem{proposition}[theorem]{Proposition}

\theoremstyle{definition}
\newtheorem*{definition}{Definition}
\newtheorem*{remark}{Remark}

\newtheorem{mainres}{Main Result}

\newcommand{\psat}{\textsc{Planar 3Sat}}
\newcommand{\indSet}{\textsc{indSet}}
\newcommand{\indQueens}{\textsc{indQueens}}
\newcommand{\indRooks}{\textsc{indRooks}}

\usepackage[hyphens]{xurl}

\tikzset{var/.pic={
code={\tikzset{scale=\figscale,transform shape}
\begin{scope}[shift={(-1,-1)}]
\foreach \x/\y in {
	4/7, 3/5, 2/3, 4/3, 5/3, 3/2, 4/2, 4/1
} {
	\path [draw=brown!, fill=brown!30] (.5+\x-0.45, .5+\y-1.45)
	-- ++(0,.9)
	-- ++(.9,0)
	-- ++(0,-.9)
	--cycle;
}
\foreach \x/\y in {
	4/8, 3/6, 2/4, 1/2
} {
	\path [draw=gray!, fill=%
	gray!30
	] (.5+\x-0.45, .5+\y-1.45)
	-- ++(0,.9)
	-- ++(.9,0)
	-- ++(0,-.9)
	--cycle;
}
\foreach \x/\y in {
	5/8, 4/6, 3/4, 2/2
} {
	\path [draw=sangria!, fill=%
	sangria!30
	] (.5+\x-0.45, .5+\y-1.45)
	-- ++(0,.9)
	-- ++(.9,0)
	-- ++(0,-.9)
	--cycle;
}
\foreach \x/\y in {5/3, 4/1} {
	\node[anchor=center, scale=1.8] at (.5+\x, \y-.5) {\queen};
}
\end{scope}
}}}

\tikzset{
ext/.pic={
code={\tikzset{scale=\figscale,transform shape}

\foreach \x/\y in {
	0/1,0/2,1/2
} {
	\path [draw=brown!, fill=brown!30] (.5+\x-0.45, .5+\y-1.45)
	-- ++(0,.9)
	-- ++(.9,0)
	-- ++(0,-.9)
	--cycle;
}
}}}

\tikzset{
inverter/.pic={
code={\tikzset{scale=\figscale,transform shape}
		\begin{scope}[shift={(-2,-2)}]
		\foreach \x/\y in {
			4/13, 3/11, 2/9, 3/9, 4/9, 5/9, 6/9, 7/9, 4/7, 3/5, 2/3
		} {
			\path [draw=brown!, fill=brown!30] (.5+\x-0.45, .5+\y-1.45)
			-- ++(0,.9)
			-- ++(.9,0)
			-- ++(0,-.9)
			--cycle;
		}
		\foreach \x/\y in {
			5/14, 4/12, 3/10, 4/8, 3/6, 2/4
		} {
			\path [draw=brown,
			 fill=
			brown!30] (.5+\x-0.45, .5+\y-1.45)
			-- ++(0,.9)
			-- ++(.9,0)
			-- ++(0,-.9)
			--cycle;
		}
		\foreach \x/\y in {
			4/14, 3/12, 2/10, 5/8, 4/6, 3/4
		} {
			\path [draw=brown,
			 fill=
			brown!30] (.5+\x-0.45, .5+\y-1.45)
			-- ++(0,.9)
			-- ++(.9,0)
			-- ++(0,-.9)
			--cycle;
		}
	\end{scope}
}}}

\tikzset{
turner/.pic={
code={\tikzset{scale=\figscale,transform shape}
		\begin{scope}[shift={(-5,-2)}]
		\foreach \x/\y in {
			2/13, 3/11, 4/9, 5/9, 7/9, 7/7, 6/5, 5/3
		} {
			\path [draw=brown!, fill=brown!30] (.5+\x-0.45, .5+\y-1.45)
			-- ++(0,.9)
			-- ++(.9,0)
			-- ++(0,-.9)
			--cycle;
		}
		
		\foreach \x/\y in {
			2/14, 3/12, 4/10, 6/10, 5/8, 7/8, 6/6, 5/4
		} {
			\path [draw=brown,
			fill=
			brown!30] (.5+\x-0.45, .5+\y-1.45)
			-- ++(0,.9)
			-- ++(.9,0)
			-- ++(0,-.9)
			--cycle;
		}
		
		\foreach \x/\y in {
			1/14, 2/12, 3/10, 6/9, 4/8, 8/8, 7/6, 6/4
		} {
			\path [draw=brown,
			fill=
			brown!30] (.5+\x-0.45, .5+\y-1.45)
			-- ++(0,.9)
			-- ++(.9,0)
			-- ++(0,-.9)
			--cycle;
		}		
	\end{scope}
}}}

\tikzset{
splitter/.pic={
code={\tikzset{scale=\figscale,transform shape}
		\begin{scope}[shift={(-2,-2)}]
		\foreach \x/\y in {
			 2/3,  10/4, 11/4, 3/5, 9/5, 4/6, 7/6, 5/7, 6/8, 7/9, 8/11, 9/13
		} {
			\path [draw=brown!, fill=brown!30] (.5+\x-0.45, .5+\y-1.45)
			-- ++(0,.9)
			-- ++(.9,0)
			-- ++(0,-.9)
			--cycle;
		}
		
		\foreach \x/\y in {
			2/4, 3/6, 8/6, 6/7, 4/8, 7/10, 8/12, 9/14, 10/5, 12/4
		} {
			\path [draw=brown,
			fill=
				brown!30] (.5+\x-0.45, .5+\y-1.45)
			-- ++(0,.9)
			-- ++(.9,0)
			-- ++(0,-.9)
			--cycle;
		}
		
		\foreach \x/\y in {
			3/4, 6/6, 4/7, 7/8, 8/10, 9/12, 10/14, 10/4, 12/3, 8/5
		} {
			\path [draw=brown,
			fill=
			brown!30] (.5+\x-0.45, .5+\y-1.45)
			-- ++(0,.9)
			-- ++(.9,0)
			-- ++(0,-.9)
			--cycle;
		}
	\end{scope}
}}}

\tikzset{
	splittersmall/.pic={
		code={\tikzset{scale=\figscale,transform shape}
			\begin{scope}[shift={(-2,-2)}]
				\foreach \x/\y in {
					2/3,  10/4, 3/5, 9/5, 4/6, 7/6, 5/7, 6/8, 7/9, 8/11
				} {
					\path [draw=brown!, fill=brown!30] (.5+\x-0.45, .5+\y-1.45)
					-- ++(0,.9)
					-- ++(.9,0)
					-- ++(0,-.9)
					--cycle;
				}
				
				\foreach \x/\y in {
					2/4, 3/6, 8/6, 6/7, 4/8, 7/10, 8/12, 10/5
				} {
					\path [draw=brown,
					 fill=
					brown!30] (.5+\x-0.45, .5+\y-1.45)
					-- ++(0,.9)
					-- ++(.9,0)
					-- ++(0,-.9)
					--cycle;
				}
				
				\foreach \x/\y in {
					3/4, 6/6, 4/7, 7/8, 8/10, 9/12, 10/4, 8/5
				} {
					\path [draw=brown,
					fill=
					brown!30] (.5+\x-0.45, .5+\y-1.45)
					-- ++(0,.9)
					-- ++(.9,0)
					-- ++(0,-.9)
					--cycle;
				}
			\end{scope}
}}}
\tikzset{
	splitterverysmall/.pic={
		code={\tikzset{scale=\figscale,transform shape}
			\begin{scope}[shift={(-2,-2)}]
				\foreach \x/\y in {
					 10/4, 3/5, 9/5, 4/6, 7/6, 5/7, 6/8, 7/9, 8/11
				} {
					\path [draw=brown!, fill=brown!30] (.5+\x-0.45, .5+\y-1.45)
					-- ++(0,.9)
					-- ++(.9,0)
					-- ++(0,-.9)
					--cycle;
				}
				
				\foreach \x/\y in {
					 3/6, 8/6, 6/7, 4/8, 7/10, 8/12, 10/5
				} {
					\path [draw=brown,
					fill=
					brown!30] (.5+\x-0.45, .5+\y-1.45)
					-- ++(0,.9)
					-- ++(.9,0)
					-- ++(0,-.9)
					--cycle;
				}
				
				\foreach \x/\y in {
				 6/6, 4/7, 7/8, 8/10, 9/12, 10/4, 8/5
				} {
					\path [draw=brown,
					fill=
					brown!30] (.5+\x-0.45, .5+\y-1.45)
					-- ++(0,.9)
					-- ++(.9,0)
					-- ++(0,-.9)
					--cycle;
				}
			\end{scope}
}}}
\tikzset{
twoclause/.pic={
code={\tikzset{scale=\figscale,transform shape}
		\begin{scope}[shift={(-2,-2)}]
	\foreach \x/\y in {
		 2/11, 2/10, 3/10, 3/9, 3/8, 4/8, 4/7, 5/7, 6/7, 3/6, 4/6, 3/5, 2/4, 3/4, 2/3
		} {
			\path [draw=brown!, fill=brown!30] (.5+\x-0.45, .5+\y-1.45)
			-- ++(0,.9)
			-- ++(.9,0)
			-- ++(0,-.9)
			--cycle;
		}
		
		\foreach \x/\y in {6/7} {
			\node[anchor=center, scale=1.8] at (.5+\x, \y-.5) {\queen};
		}
		
	\end{scope}
}}}

\tikzset{
threeclause/.pic={
code={\tikzset{scale=\figscale,transform shape}
	\begin{scope}[shift={(-2,-3)}]
		\foreach \x/\y in {
			2/12,  12/12,
			2/11, 3/11, 11/11, 12/11,
			3/10, 11/10,
			3/9, 4/9, 6/9, 7/9, 8/9, 10/9, 11/9,
			4/8, 5/8, 6/8, 8/8, 9/8, 10/8,
			3/7, 4/7, 5/7, 9/7, 10/7,
			3/6, 10/6, 11/6, 12/6,
			2/5, 3/5, 11/5,
			2/4,  11/4, 12/4,
			9/3, 10/3, 12/3,
			10/2, 11/2, 12/2, 13/2,
			10/1
		} {
			\path [draw=brown!, fill=brown!30] (.5+\x-0.45, .5+\y-1.45)
			-- ++(0,.9)
			-- ++(.9,0)
			-- ++(0,-.9)
			--cycle;
		}
		
		\foreach \x/\y in { 12/6, 11/4, 9/3,  12/2, 10/1} {
			\node[anchor=center, scale=1.8] at (.5+\x, \y-.5) {\queen};
		}
		
		\end{scope}
}}}

\tikzset{
fourclause/.pic={
code={\tikzset{scale=\figscale,transform shape}
		\begin{scope}[shift={(-2,-2)}]
		\foreach \x/\y in {
			2/11, 12/11,
			2/10, 3/10, 11/10, 12/10,
			3/9, 11/9,
			3/8, 4/8, 6/8, 7/8, 8/8, 10/8, 11/8,
			4/7, 5/7, 6/7, 8/7, 9/7, 10/7,
			3/6, 4/6, 5/6, 9/6, 10/6, 11/6,
			3/5, 11/5,
			2/4, 3/4, 11/4, 12/4,
			2/3, 12/3
		} {
			\path [draw=brown!, fill=brown!30] (.5+\x-0.45, .5+\y-1.45)
			-- ++(0,.9)
			-- ++(.9,0)
			-- ++(0,-.9)
			--cycle;
		}
		\end{scope}
}}}

\tikzset{
	leftshift/.pic={
		code={\tikzset{scale=\figscale,transform shape}
			
			\foreach \x/\y in {
				0/0,
				0/1, 1/1, 3/1,
				1/2,2/2,3/2,4/2,
				0/3,1/3,2/3,4/3,
				1/4,
				1/5,2/5
			} {
				\path [draw=brown!, fill=brown!30] (.5+\x-0.45, .5+\y-1.45)
				-- ++(0,.9)
				-- ++(.9,0)
				-- ++(0,-.9)
				--cycle;
			}
}}}

\tikzset{
	forwardshift/.pic={
		code={\tikzset{scale=\figscale,transform shape}
			
			\foreach \x/\y in {
				1/0,
				1/1, 2/1,
				2/2,4/2,
				2/3,3/3,4/3,
				3/4,4/4,
				3/5,4/5,5/5,
				4/6,5/6,
				5/7,
				5/8,6/8,
				6/9,
				6/10,7/10,9/10,
				7/11,8/11,9/11,10/11,
				6/12,7/12,8/12,10/12,
				7/13,
				7/14,8/14
			} {
				\path [draw=brown!, fill=brown!30] (.5+\x-0.45, .5+\y-1.45)
				-- ++(0,.9)
				-- ++(.9,0)
				-- ++(0,-.9)
				--cycle;
			}
}}}

\tikzset{
	rightshift/.pic={
		code={\tikzset{scale=\figscale,transform shape}
			
			\foreach \x/\y in {
				0/0,1/0,
				1/1, 2/1, 3/1, 4/1,
				3/2,
				1/3,2/3,3/3,4/3,5/3,
				1/4,5/4,6/4
			} {
				\path [draw=brown!, fill=brown!30] (.5+\x-0.45, .5+\y-1.45)
				-- ++(0,.9)
				-- ++(.9,0)
				-- ++(0,-.9)
				--cycle;
			}
}}}

\begin{document}
	
	\title{Maximum independent 
		queen set on polyominoes is NP-complete}
	\author[A. Langlois-Rémillard]{Alexis Langlois-Rémillard}
	\address[
	A. Langlois-Rémillard]{Hausdorff Center for Mathematics, Endenicher Allee 60, 53115, Bonn, Germany}
	\email{langlois@uni-bonn.de; \href{https://alexisl-r.github.io/}{https://alexisl-r.github.io/};  \href{https://orcid.org/0000-0002-5919-8766}{ORCID:0000-0002-5919-8766}}
	\author[M. Müßig]{Mia M\"ußig}
	\address[M.  Müßig]{Ludwig Maximilian Universit\"at M\"unchen, Germany}
	\email{nienna@miamuessig.de; \href{https://orcid.org/0000-0003-0249-6831}{ORCID:0000-0003-0249-6831}}

	\begin{abstract}
		Finding a set of vertices in a graph with no edges between them, \indSet, is a well-known NP-complete problem. The queen graph of a chessboard is constructed by taking vertices as the tiles of the chessboard and drawing edges between two tiles if a queen can move from one to the other. We call \indQueens{} the independent set problem on a queen graph where the chessboard is a polyomino. We prove that \indQueens{} on polyominoes is NP-complete, proving a conjecture of Langlois-Rémillard--Mü\ss{}ig--Roldán. As our reduction is parsimonious, we can further prove that it is \#P-complete. We furthermore prove that \indRooks{} on polyominoes is \#P-complete, despite being solvable in polynomial time.
	\end{abstract}
		\let\MakeUppercase \relax 
	\begingroup 
	\maketitle
	\endgroup
	The independent set problem for graphs, \indSet{}, asks, given a graph $G$ and a number $k$, to find a set of $k$ vertices in $G$ with no edges between them. It is a famous NP-complete problem equivalent to finding a clique in a graph~\cite{Karp72}. Its inception can be traced back to the famous 8-queens problem, first posed by Max Bezzel in 1848~\cite{Bezzel48}: \emph{how many ways are there to place 8 queens on an $8\times 8$ chessboard so that no two queens attack each other?} For an account of the history, see~\cite{Campbell77}. We return to the chess origin of \indSet{} and consider a subset of graphs for which finding an independent set is equivalent to placing a number of non-attacking queens on generalised chessboards. Having this correspondence in mind, we will sometime call non-attacking queens ``independent'' queens.
	
	The first generalisation of the $8$-queens puzzle is to consider bigger boards, namely, the $n$-queens problem asking to place $n$ non-attacking queens on an $n\times n$ chessboard. It has been studied for a long time; with one of its first mathematical studies, Pauls' formulas for exhibiting solutions given any $n\times n$ chessboards, dating back to 1874~\cite{Pauls74}. Finding \emph{one} solution is thus a trivial endeavour at the level of complexity, since we have closed-form formul\ae{} for them. On the other hand, the completion problem, asking whether it is possible to place $n-k$ independent queens starting from a given position of $k$ queens on an $n\times n$ chessboard, has been shown to be NP-complete~\cite{GJN17}. Finally, counting the number of solutions is also much harder than finding one. It still provides interesting avenues of research in computer science and combinatorics~\cite{BS09}, and there have been exciting new developments on bounds in recent years~\cite{Luria17,LS22,BK21,GMS22,Simkin23,MITHardness25}.

	There have been other interesting generalisations of the $n$-queen problem. One can change the geometry of the board, for example, considering it on toroidal chessboards~\cite{Polya1918,BK21}, on other plane tessellations~\cite{TB2000}, by adding walls on the chessboard~\cite{Barnaby07}, or on higher-dimensional chessboards~\cite{Nudelman96,BS06,Kunt24,Ramani26}. 
	
	In this work, following the line of~\cite{AR21,LRMR25}, we instead focus on considering  polyomino chessboards, so finite connected subsets of the regular square tiling~\cite{Golomb54}. We are interested in the computational complexity of \indSet{} with those polyomino graphs.  
	
	In \cite{LRMR25}, the authors and Érika Roldán proved that the independent queen set problem and the independent rook set problem were NP-complete on polycubes~\cite[Theorems~1--2]{LRMR25}. The independent rook set problem on polyominoes was proven to be in P by Hannah Alpert and Érika Roldán~\cite[Theorem~12]{AR21}. Until now, it was open whether the independent queen set problem on polyominoes is NP-complete.
	
	\subsection*{Results}	
	We prove that the independent queen set problem on polyominoes is NP-complete, proving Conjecture~17 of~\cite{LRMR25}. Combined with~\cite[Thm.~1]{LRMR25}, we thus have that independent queen set on $d$-polycubes for $d\geq 2$ is NP-complete.
	
	\begin{mainres}[Theorem~\ref{mainres:NPQueenPolyo}]\label{mainres:NPQueenPoly}
		Independent queen set on polyominoes is NP-complete. 
	\end{mainres}
	
	This result is proven through a reduction to \psat. In fact, it is even a \emph{parsimonious reduction}, that is, a reduction that preserves the number of solutions. We will use this fact to prove two additional theorems of own interest on the \#P-completeness of independent queen and rook set problems.
	
	\begin{mainres}[Theorems~\ref{thm:queenCountP} and~\ref{thm:rookCountP}]\label{mainres:counting}
		Counting the number of maximum solutions of maximum independent rook and queen sets on polyominoes is \#P-complete.
	\end{mainres}

	It might be of interest to the reader how the specific gadgets were found. In contrast to~\cite{LRMR25}, they were not designed by hand but instead obtained from a computer search. Note that while a full description of what makes a good gadget is often very hard, we can still write down many concrete characteristics that can be checked programmatically. For example, we can test if a given polyomino with $m\leq 4$ tiles can be glued one to four times at a top right edge of another polyomino with $n\leq 8$ tiles, such that in all steps the constructed figure has exactly two maximum placements of non-attacking queens and those maximum queen placements differ in all parts of the figure. A human then only needs to choose from a small list of candidates for the base and extension gadgets.

	\subsection*{Structure of the paper}	
	In Section~\ref{sec:prelim}, we give the definitions of the problems we consider and the basic tools we will be using throughout. In Section~\ref{sec:reduction}, we construct the reduction to \psat{} through a set of lemmata constructing the gadgets. We then present results on the \#P-complexity of \indQueens{} and \indRooks{} in Section~\ref{sec:counting}. Finally, we present a complete example of the reduction in Appendix~\ref{sec:ex}.¸
	
	\subsection*{Disclosure of the use of AI} 
	The various brute-force implementations for generating the respective gadget candidates were written with the help of Gemini~3.5 Flash, and all verifications were done with human-written code calling CPLEX~22.1. In particular, the AI did not design the gadgets, only wrote the brute-force program to produce candidates. Gemini~3.5 Flash was additionally used to help visualise the final gadgets with TikZ. 
	Our web simulator~\cite{LRM26Software} implementing our solver was also produced with the assistance of Gemini~3.5 Flash and Claude Sonnet~4.6. No AI was used in the actual proof nor in the redaction.

	
	\section{Preliminaries}\label{sec:prelim}
	We first begin by defining the problem at hand. For this, we define the movement of a queen on a polyomino to be the same movement as a chess queen (horizontal, vertical and diagonal lines), not allowing jumping through holes. Figure~\ref{fig:mov_queen} gives an example.
	
	\begin{figure}[h]
		\begin{tikzpicture}[scale=.5]
		\foreach \x/\y in {-1/-1, 0/-1, 1/-1, 2/-1, 2/0, 0/1, -1/1, -2/1, 2/1, 0/2, 1/2, 2/2, -2/2,2/3, 0/3, -1/3, -2/3, 0/4, 1/4, -1/4, -2/4, 0/5, 1/5, 2/5} { 
			\path [draw=brown, thick, fill=brown!30] (.5+\x-0.45, .5+\y-0.45)
			-- ++(0,.9)
			-- ++(.9,0)
			-- ++(0,-.9)
			--cycle;
		}
		
		\foreach \x/\y in {0/1,  0/2,  0/4,  0/5,-1/3,-2/3,-1/4,1/4,2/5,1/2,2/1} 
		{
			\path [draw=teal, fill=teal!70, pattern=crosshatch, pattern color = teal] (.5+\x-0.45, .5+\y-0.45)
			-- ++(0,.9)
			-- ++(.9,0)
			-- ++(0,-.9)
			--cycle;
		}
		\draw[ultra thick] (-1,-1) -- (3,-1) -- (3,4) --(2,4)-- (2,3)-- (1,3) -- (1,4) - -(2,4) -- (2,5) -- (3,5) -- (3,6) -- (0,6) -- (0,5) -- (-2,5) -- (-2,1) -- (1,1) -- (1,2) -- (2,2) -- (2,0) -- (-1,0)--cycle ;
		\draw[ultra thick] (-1,2) rectangle (0,3);
		\draw[teal] (0.5, 3.5) node {\queen};
		\end{tikzpicture} \hspace{3cm}
		\begin{tikzpicture}[scale=.5]
		\foreach \x/\y in {-1/-1, 0/-1, 1/-1, 2/-1, 2/0, 0/1, -1/1, -2/1, 2/1, 0/2, 1/2, 2/2, -2/2,2/3, 0/3, -1/3, -2/3, 0/4, 1/4, -1/4, -2/4, 0/5, 1/5, 2/5} { 
			\path [draw=brown, thick, fill=brown!30] (.5+\x-0.45, .5+\y-0.45)
			-- ++(0,.9)
			-- ++(.9,0)
			-- ++(0,-.9)
			--cycle;
		}
		
		\foreach \x/\y in {0/1,  0/2,  0/4,  0/5,-1/3,-2/3} 
		{
			\path [draw=teal, fill=teal!70, pattern=crosshatch, pattern color = teal] (.5+\x-0.45, .5+\y-0.45)
			-- ++(0,.9)
			-- ++(.9,0)
			-- ++(0,-.9)
			--cycle;
		}
		\draw[ultra thick] (-1,-1) -- (3,-1) -- (3,4) --(2,4)-- (2,3)--  (1,3) -- (1,4) - -(2,4) -- (2,5) -- (3,5) -- (3,6) -- (0,6) -- (0,5) -- (-2,5) -- (-2,1) -- (1,1) -- (1,2) -- (2,2) -- (2,0) -- (-1,0)--cycle ;
		\draw[ultra thick] (-1,2) rectangle (0,3);
		\draw[teal] (0.5, 3.5) node {\rook};
		\end{tikzpicture}
		\caption{Left: A queen and the tiles (in teal and crosshatched) she can reach. Right: A rook and the tiles (in teal and crosshatched) it can reach.}\label{fig:mov_queen}
	\end{figure}
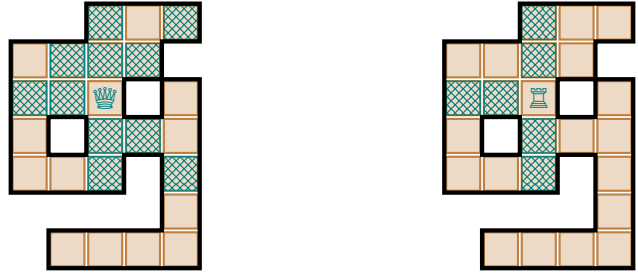
	
	\begin{definition}[\indQueens{} and \indRooks
		] 
		We say that an instance of the independent queen set problem for polyominoes is a pair $(P,m)^Q$, where $P$ is a polyomino and $m$ is a positive integer. The problem asks whether there exists a non-attacking configuration of $m$ queens placed on $P$. We denote it \indQueens. 
		
		Similarly, an instance of the independent rook set problem  is a pair $(P,m)^R$ asking whether there exists a non-attacking configuration of $m$ rooks on $P$, and we denote it \indRooks.
	\end{definition}
	
	The reduction we will use will be to a classical NP-complete problem.
	
	\begin{definition}[\psat~\cite{Lichtenstein82}]\label{def:psat}
		Given a set of Boolean variables $x_i$ that satisfy a set of clauses of the form $x_{i_1}\lor x_{i_2}$ or $x_{i_1}\lor x_{i_2} \lor x_{i_3}$, with each $x_{i_k}$ appearing being the literal $x_{i_k}$ or $\overline{x}_{i_k}$, we construct a bipartite graph with the two sets of vertices given by the variables and the clauses, and with the set of edges given by linking a clause $c$ and a variable $x$ if the clause $c$ contains either literals $x$ or $\overline{x}$. The problem of deciding whether there exists a truth assignment to the variables such that all clauses are satisfied is called $\textsc{3sat}$. If the bipartite graph constructed is planar, this problem is instead called \textsc{planar 3sat}. 
	\end{definition}
	
	An example of such an instance is given in Figure~\ref{fig:ex_psat}.

	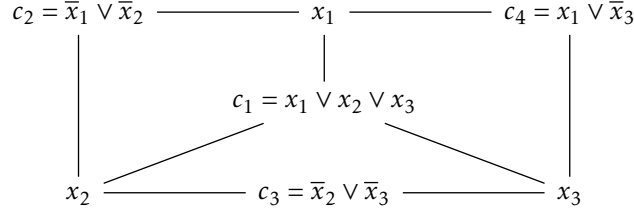
\begin{figure}[h]
		\centering
		\begin{tikzcd}
		c_2=\overline{x}_1\lor \overline{x}_2 \arrow[r, dash] \arrow[dd,dash] & x_1 \arrow[d,dash] \arrow[r,dash] & c_4= x_1 \lor \overline{x}_3 \arrow[dd, dash]\\
		& c_1=x_1 \lor x_2 \lor x_3 \arrow[ld,dash] \arrow[dr,dash]& \\
		x_2 \arrow[r,dash] & c_3=\overline x_2\lor \overline x_3 & \arrow[l,dash] x_3
		\end{tikzcd}
		\caption{An instance $C$ of $\psat$ with three variables, $x_1$, $x_2$ and $x_3$, and four clauses, $c_1=x_1\lor x_2\lor x_3$, $c_2=\overline x_1\lor \overline x_2$, $c_3=\overline x_2\lor \overline x_3$ and $c_4=x_1\lor \overline x_3$. There are three solutions realising the instance: $(x_1,x_2,x_3) = (\mathsf{T},\mathsf F,\mathsf T),\ (\mathsf T,\mathsf F,\mathsf F),\ (\mathsf F,\mathsf T,\mathsf F)$.}\label{fig:ex_psat} 
	\end{figure}
	
	\begin{proposition}[\cite{Lichtenstein82}]\label{prop:psatNP}
		The problem \psat{} is NP-complete.
	\end{proposition}
	
	A final key tool for the proof of the main theorem is the so-called visibility graph associated with a planar graph~\cite{DHLVM83}, which enables us to transfer a planar graph on a grid representation~\cite{Biedl13}. An example of a grid representation is shown in Figure~\ref{fig:path_graph-psat}. This is also a key idea of~\cite{AR21,LRMR25}.

	\begin{proposition}[{\cite[Theorem~4]{Biedl13}}]\label{prop:gridpath}
		It is possible to obtain a polynomial size grid path graph of a planar graph in polynomial time.
	\end{proposition}
	
	\begin{figure}[h]
		\centering
		\begin{tikzpicture}[scale=.52]
		\foreach \x/\y in {0/0,0/1,0/2,0/3,0/4,0/5,0/6,0/7,0/8,
			1/0,1/4,1/8,
			2/0,2/1,2/2,2/3,2/4,2/8,
			3/0,3/4,3/8,
			4/0,4/4,4/5,4/5,4/6,4/7,4/8,
			5/0,5/8,
			6/0,6/1,6/2,6/3,6/4,6/5,6/6,6/7,6/8
		} { 
			\path [draw=gray!, fill=brown!50] (.5+\x-0.45, .5+\y-1.45)
			-- ++(0,.9)
			-- ++(.9,0)
			-- ++(0,-.9)
			--cycle;
		}
		\foreach \x/\y in {0/0,4/4,0/8}
		{
			\path [draw=gray!, fill=teal!50] (.5+\x-0.45, .5+\y-1.45)
			-- ++(0,.9)
			-- ++(.9,0)
			-- ++(0,-.9)
			--cycle;
		}
		\foreach \x/\y in {0/4,2/2,4/6,6/4}
		{
			\path [draw=gray!, fill=purple!50] (.5+\x-0.45, .5+\y-1.45)
			-- ++(0,.9)
			-- ++(.9,0)
			-- ++(0,-.9)
			--cycle;
		}
		\node[anchor=west] at (.5+0-0.45, 0-.52) {$x_1$};
		\node[anchor=west] at (.5+4-0.45, 4-.52) {$x_2$};
		\node[anchor=west] at (.5+0-0.45, 8-.52) {$x_3$};
		\node[anchor=west] at (.5+0-0.45, 4-.52) {$c_1$};
		\node[anchor=west] at (.5+2-0.45, 2-.52) {$c_2$};
		\node[anchor=west] at (.5+4-0.45, 6-.52) {$c_3$};
		\node[anchor=west] at (.5+6-0.45, 4-.52) {$c_4$};
		\end{tikzpicture}
		\caption{The grid path graph of the instance $C$ of Figure~\ref{fig:ex_psat}. The clauses are in purple, the variables are in teal and the connection squares are in brown.}
		\label{fig:path_graph-psat}
	\end{figure}
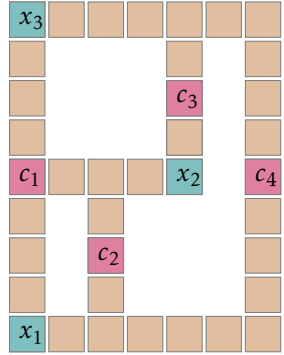
	
	\section{Maximum independent queens is NP-complete on polyominoes}
	\label{sec:reduction}
	
	The goal of this section is to prove Main Result~\ref{mainres:NPQueenPoly}. 
	
	\begin{theorem}[{\cite[Conjecture~17]{LRMR25}}]\label{mainres:NPQueenPolyo}
		The problem \indQueens{} is NP-complete.
	\end{theorem}
	
	The section will prove this theorem by a series of lemmata. First we show that verifying a solution to be done in polynomial time. Second we give a polynomial-size reduction of a known NP-complete problem, here \psat{}, to \indQueens.

	We will prove that the gadgets indeed provide a translation between the independent problem and SAT, and that the translation is of polynomial growth.
	
	\begin{lemma}\label{lem:Verif}
		The  independent queen set problem on polyominoes is verifiable in polynomial time.
	\end{lemma}
	\begin{proof}
		This is equivalent to the problem of finding a coclique on a graph, which is known to be verifiable polynomially; see, for example,~\cite{PapadimitriouSteiglitz98}.
	\end{proof}

	We now state five lemmas concerning gadgets. All the statements of Lemmata~\ref{lem:GadgetBase}--\ref{lem:Gadget3Clause} are verified using the solver~\cite{QandRsoftware}, and the readers can test for themselves with the user-friendly program~\cite{LRM26Software}.

	\begin{lemma}\label{lem:GadgetBase}
		The polyomino presented in Figure~\ref{fig:base-gadget} is a variable gadget, which encodes the two instances of the literal $x$.
	\end{lemma}
	\begin{proof}
		The base gadget of Figure~\ref{fig:base-gadget} with any number $N$ of extensions has two maximum placements of $N+4$ queens  differing by whether $N+2$ of them are on the gray `\textsf{T}' or the red `\textsf{F}' tiles.
	\end{proof}

	\begin{figure}[htbp]
		\begin{tabular}{cc}
			\begin{tabular}[b]{c}
				\begin{subfigure}[t]{.4\textwidth}
					\centering
					\begin{tikzpicture}[scale=.5,transform shape]
					\foreach \x/\y in {
						4/7, 3/5, 2/3, 4/3, 5/3, 3/2, 4/2, 4/1
					} {
						\path [draw=brown!, fill=brown!30] (.5+\x-0.45, .5+\y-1.45)
						-- ++(0,.9)
						-- ++(.9,0)
						-- ++(0,-.9)
						--cycle;
					}
					
					\foreach \x/\y in {
						4/8, 3/6, 2/4, 1/2
					} {
						\path [draw=gray!, fill=gray!30] (.5+\x-0.45, .5+\y-1.45)
						-- ++(0,.9)
						-- ++(.9,0)
						-- ++(0,-.9)
						--cycle;
						\node[anchor=center, scale=1.5] at (.5+\x, \y-.5) {\textbf{\textsf{T}}};
					}
					
					\foreach \x/\y in {
						5/8, 4/6, 3/4, 2/2
					} {
						\path [draw=sangria!, fill=sangria!30] (.5+\x-0.45, .5+\y-1.45)
						-- ++(0,.9)
						-- ++(.9,0)
						-- ++(0,-.9)
						--cycle;
						\node[anchor=center, scale=1.5] at (.5+\x, \y-.5) {\textbf{\textsf{F}}};
					}
					
					\foreach \x/\y in {5/3, 4/1} {
						\node[anchor=center, scale=1.8] at (.5+\x, \y-.5) {\queen};
					}
					
					\draw[line width=1.6pt, teal] (1,1) -- (4,1) -- (4,0) -- (5,0) -- (5,2) -- (6,2) -- (6,3) -- (4,3) -- (4,2) -- (1,2) -- cycle; 
					\draw[line width=1.6pt, teal] (2,2) -- (3,2) -- (3,3) -- (4,3) -- (4,4) -- (2,4) -- cycle; 
					\draw[line width=1.6pt, black] (3,4) -- (4,4) -- (4,5) -- (5,5) -- (5,6) -- (3,6) -- cycle; 
					\draw[line width=1.6pt, black] (4,6) -- (5,6) -- (5,7) -- (6,7) -- (6,8) -- (4,8) -- cycle; 
					\end{tikzpicture}
					\caption[Base gadget (highlighted in teal) with 2 extensions. The extensions L can go on.]{Base gadget with 2 extensions. Any number of extensions (Figure~\ref{fig:extension}) can be added.}
					\label{fig:base-gadget}
				\end{subfigure}\\
				\\
				\begin{subfigure}[b]{.4\textwidth}
					\centering
					\begin{tikzpicture}[scale=.5]
					\draw[ultra thick] (0,0) -- (1,0) -- (1,1) -- (2,1)--(2,2)--(0,2)--cycle;
					\foreach \x/\y in {
						0/1,0/2,1/2
					} {
						\path [draw=brown!, fill=brown!30] (.5+\x-0.45, .5+\y-1.45)
						-- ++(0,.9)
						-- ++(.9,0)
						-- ++(0,-.9)
						--cycle;
					}
					\end{tikzpicture}
					\caption{An extension.}\label{fig:extension}
				\end{subfigure}
			\end{tabular}
			& 
			\begin{subfigure}[b]{.45\textwidth}
				\centering
				\begin{tikzpicture}[scale=.5,transform shape]
				\foreach \x/\y in {
					4/13, 3/11, 2/9, 3/9, 4/9, 5/9, 6/9, 7/9, 4/7, 3/5, 2/3, 4/3, 5/3, 3/2, 4/2, 4/1
				} {
					\path [draw=brown!, fill=brown!30] (.5+\x-0.45, .5+\y-1.45)
					-- ++(0,.9)
					-- ++(.9,0)
					-- ++(0,-.9)
					--cycle;
				}
				
				\foreach \x/\y in {
					5/14, 4/12, 3/10, 4/8, 3/6, 2/4, 1/2
				} {
					\path [draw=gray!, fill=gray!30] (.5+\x-0.45, .5+\y-1.45)
					-- ++(0,.9)
					-- ++(.9,0)
					-- ++(0,-.9)
					--cycle;
					\node[anchor=center, scale=1.5] at (.5+\x, \y-.5) {\textbf{\textsf{T}}};
				}
				
				\foreach \x/\y in {
					4/14, 3/12, 2/10, 5/8, 4/6, 3/4, 2/2
				} {
					\path [draw=sangria!, fill=sangria!30] (.5+\x-0.45, .5+\y-1.45)
					-- ++(0,.9)
					-- ++(.9,0)
					-- ++(0,-.9)
					--cycle;
					\node[anchor=center, scale=1.5] at (.5+\x, \y-.5) {\textbf{\textsf{F}}};
				}
				
				\foreach \x/\y in {7/9, 5/3, 4/1} {
					\node[anchor=center, scale=1.8] at (.5+\x, \y-.5) {\queen};
				}
				
				\draw[line width=1.6pt, black] (1,1) -- (4,1) -- (4,0) -- (5,0) -- (5,2) -- (6,2) -- (6,3) -- (4,3) -- (4,2) -- (1,2) -- cycle; 
				\draw[line width=1.6pt, black] (2,2) -- (3,2) -- (3,3) -- (4,3) -- (4,4) -- (2,4) -- cycle; 
				\draw[line width=1.6pt, black] (3,4) -- (4,4) -- (4,5) -- (5,5) -- (5,6) -- (3,6) -- cycle; 
				\draw[line width=1.6pt, teal] (4,6) -- (5,6) -- (5,7) -- (6,7) -- (6,8) -- (4,8) -- cycle; 
				\draw[line width=1.6pt, teal] (3,8) -- (8,8) -- (8,9) -- (3,9) -- cycle; 
				\draw[line width=1.6pt, teal] (2,8) -- (3,8) -- (3,9) -- (4,9) -- (4,10) -- (2,10) -- cycle; 
				\draw[line width=1.6pt, black] (3,10) -- (4,10) -- (4,11) -- (5,11) -- (5,12) -- (3,12) -- cycle; 
				\draw[line width=1.6pt, black] (4,12) -- (5,12) -- (5,13) -- (6,13) -- (6,14) -- (4,14) -- cycle; 
				\end{tikzpicture}
				\caption{Inverter gadget in teal placed on a base gadget.}
				\label{fig:inverter-gadget}
			\end{subfigure}
		\end{tabular}
		\caption{Two gadgets for variable, encoding $x$ or $\overline{x}$, and an extension piece used to enlarge all gadgets. The inverter gadget of Figure~\ref{fig:inverter-gadget} can also be placed across any extension to change a variable to its inverse.}
	\end{figure}
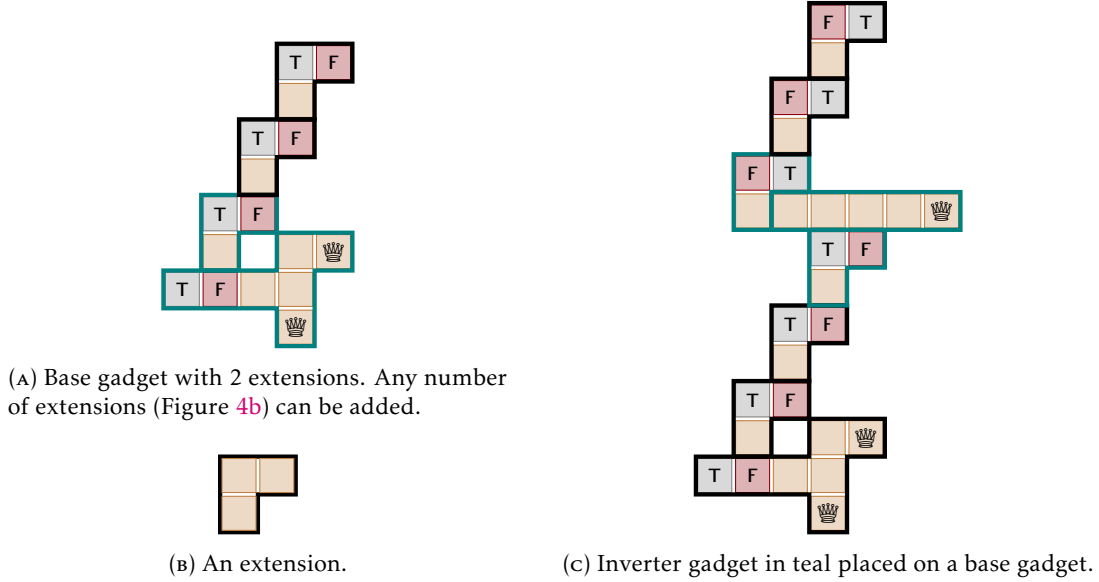
	
	\begin{lemma}\label{lem:GadgetInverter}
		The gadget of Figure~\ref{fig:inverter-gadget} encodes the literal $\overline{x}$.
	\end{lemma}
	\begin{proof}
		The inverter gadget can be added to a base gadget and it will also have two maximum queen placements are on the  gray `\textsf{T}' or the red `\textsf{F}' tiles, but it inverts the role of gray `\textsf{T}' or red `\textsf{F}' tiles after it passes. Thus a base gadget with maximum queen placement related to value \texttt{T} will be read as an extension with maximum queen placement related to value \texttt{F}, and vice versa, as can be seen in Figure~\ref{fig:inverter-gadget}.
	\end{proof}
	\begin{lemma}\label{lem:GadgetTurning}
		The gadget of Figure~\ref{fig:diagonal-left-turner} enables us to reflect the signal. It has two maximum queen placements of 11 queens, differing by whether 9 queens occupy the gray `\textsf{T}' or the red `\textsf{F}' tiles.
	\end{lemma}
	\begin{proof}
		The two maximum queen placements can be read of the tiles of Figure~\ref{fig:diagonal-left-turner}. Any number of extensions can be added without changing the fact that there will only be two solutions.
	\end{proof}

	\begin{figure}[htbp]
		
		\centering
		\begin{subfigure}{.45\textwidth}
			\centering 
			\begin{tikzpicture}[scale=.5,transform shape]
			\foreach \x/\y in {
				2/13, 3/11, 4/9, 5/9, 7/9, 7/7, 6/5, 5/3, 7/3, 6/2, 7/2
			} {
				\path [draw=brown!, fill=brown!30] (.5+\x-0.45, .5+\y-1.45)
				-- ++(0,.9)
				-- ++(.9,0)
				-- ++(0,-.9)
				--cycle;
			}
			
			\foreach \x/\y in {
				2/14, 3/12, 4/10, 6/10, 5/8, 7/8, 6/6, 5/4, 4/2
			} {
				\path [draw=gray!, fill=gray!30] (.5+\x-0.45, .5+\y-1.45)
				-- ++(0,.9)
				-- ++(.9,0)
				-- ++(0,-.9)
				--cycle;
				\node[anchor=center, scale=1.5] at (.5+\x, \y-.5) {\textbf{\textsf{T}}};
			}
			
			\foreach \x/\y in {
				1/14, 2/12, 3/10, 6/9, 4/8, 8/8, 7/6, 6/4, 5/2
			} {
				\path [draw=sangria!, fill=sangria!30] (.5+\x-0.45, .5+\y-1.45)
				-- ++(0,.9)
				-- ++(.9,0)
				-- ++(0,-.9)
				--cycle;
				\node[anchor=center, scale=1.5] at (.5+\x, \y-.5) {\textbf{\textsf{F}}};
			}
			
			\foreach \x/\y in {8/3, 7/1} {
				\path [draw=brown!, fill=brown!30] (.5+\x-0.45, .5+\y-1.45)
				-- ++(0,.9)
				-- ++(.9,0)
				-- ++(0,-.9)
				--cycle;
				\node[anchor=center, scale=1.8] at (.5+\x, \y-.5) {\queen};
			}
			
			\draw[line width=1.6pt, black] (4,1) -- (7,1) -- (7,0) -- (8,0) -- (8,2) -- (9,2) -- (9,3) -- (7,3) -- (7,2) -- (4,2) -- cycle; 
			\draw[line width=1.6pt, black] (5,2) -- (6,2) -- (6,3) -- (7,3) -- (7,4) -- (5,4) -- cycle; 
			\draw[line width=1.6pt, black] (6,4) -- (7,4) -- (7,5) -- (8,5) -- (8,6) -- (6,6) -- cycle; 
			\draw[line width=1.6pt, black] (2,11) -- (2,12) -- (4,12) -- (4,10) -- (3,10) -- (3,11) -- cycle; 
			\draw[line width=1.6pt, black] (1,13) -- (1,14) -- (3,14) -- (3,12) -- (2,12) -- (2,13) -- cycle; 
			\draw[line width=1.6pt, teal] (7,6) -- (8,6) -- (8,7) -- (9,7) -- (9,8) -- (7,8) -- cycle; 
			\draw[line width=1.6pt, teal] (4,7) -- (6,7) -- (6,8) -- (8,8) -- (8,9) -- (7,9) -- (7,10) -- (6,10) -- (6,9) -- (5,9) -- (5,8) -- (4,8) -- cycle;
			\draw[line width=1.6pt, teal] (3,9) -- (3,10) -- (5,10) -- (5,8) -- (4,8) -- (4,9) -- cycle; 
			\end{tikzpicture}
			\caption{Turning gadget on top of a base gadget.}
			\label{fig:diagonal-left-turner}
		\end{subfigure}
		\begin{subfigure}{.5\textwidth}
			\centering
			\begin{tikzpicture}[scale=.5,transform shape]
			\foreach \x/\y in {
				4/1, 5/3, 3/2, 4/2, 2/3, 4/3, 10/4, 11/4, 3/5, 9/5, 4/6, 7/6, 5/7, 6/8, 7/9, 8/11, 9/13
			} {
				\path [draw=brown!, fill=brown!30] (.5+\x-0.45, .5+\y-1.45)
				-- ++(0,.9)
				-- ++(.9,0)
				-- ++(0,-.9)
				--cycle;
			}
			
			\foreach \x/\y in {
				1/2, 2/4, 3/6, 8/6, 6/7, 4/8, 7/10, 8/12, 9/14, 10/5, 12/4
			} {
				\path [draw=gray!, fill=gray!30] (.5+\x-0.45, .5+\y-1.45)
				-- ++(0,.9)
				-- ++(.9,0)
				-- ++(0,-.9)
				--cycle;
				\node[anchor=center, scale=1.5] at (.5+\x, \y-.5) {\textbf{\textsf{T}}};
			}
			
			\foreach \x/\y in {
				2/2, 3/4, 6/6, 4/7, 7/8, 8/10, 9/12, 10/14, 10/4, 12/3, 8/5
			} {
				\path [draw=sangria!, fill=sangria!30] (.5+\x-0.45, .5+\y-1.45)
				-- ++(0,.9)
				-- ++(.9,0)
				-- ++(0,-.9)
				--cycle;
				\node[anchor=center, scale=1.5] at (.5+\x, \y-.5) {\textbf{\textsf{F}}};
			}
			
			\foreach \x/\y in {4/1, 5/3} {
				\node[anchor=center, scale=1.8] at (.5+\x, \y-.5) {\queen};
			}
			
			\draw[line width=1.6pt, black] (1,1) -- (4,1) -- (4,0) -- (5,0) -- (5,2) -- (6,2) -- (6,3) -- (4,3) -- (4,2) -- (1,2) -- cycle; 
			\draw[line width=1.6pt, black] (2,2) -- (3,2) -- (3,3) -- (4,3) -- (4,4) -- (2,4) -- cycle; 
			\draw[line width=1.6pt, black] (7,8) -- (8,8) -- (8,9) -- (9,9) -- (9,10) -- (7,10) -- cycle; 
			\draw[line width=1.6pt, black] (8,10) -- (9,10) -- (9,11) -- (10,11) -- (10,12) -- (8,12) -- cycle; 
			\draw[line width=1.6pt, black] (9,12) -- (10,12) -- (10,13) -- (11,13) -- (11,14) -- (9,14) -- cycle; 
			\draw[line width=1.6pt, black] (8,4) -- (9,4) -- (9,6) -- (7,6) -- (7,5) -- (8,5) -- cycle; 
			\draw[line width=1.6pt, black] (10,3) -- (11,3) -- (11,5) -- (9,5) -- (9,4) -- (10,4) -- cycle; 
			\draw[line width=1.6pt, black] (12,2) -- (13,2) -- (13,4) -- (11,4) -- (11,3) -- (12,3) -- cycle; 
			\draw[line width=1.6pt, teal] (3,4) -- (4,4) -- (4,5) -- (5,5) -- (5,6) -- (3,6) -- cycle; 
			\draw[line width=1.6pt, teal] (6,5) -- (7,5) -- (7,6) -- (7,7) -- (8,7) -- (8,8) -- (6,8) -- (6,7) -- (5,7) -- (5,8) -- (4,8) -- (4,6) -- (6,6) -- cycle; 
			\end{tikzpicture}
			\caption{Splitter gadget on top of a base gadget. It also allows right turn, and we achieve left turn by simply right turning 3 times.}
			\label{fig:splitter-gadget}
		\end{subfigure}
		\caption{The two gadgets are there to move a variable along. They are also used to turn and reflect in order to have all potential orientations. In both of them, any number of extensions can be added.}
	\end{figure}
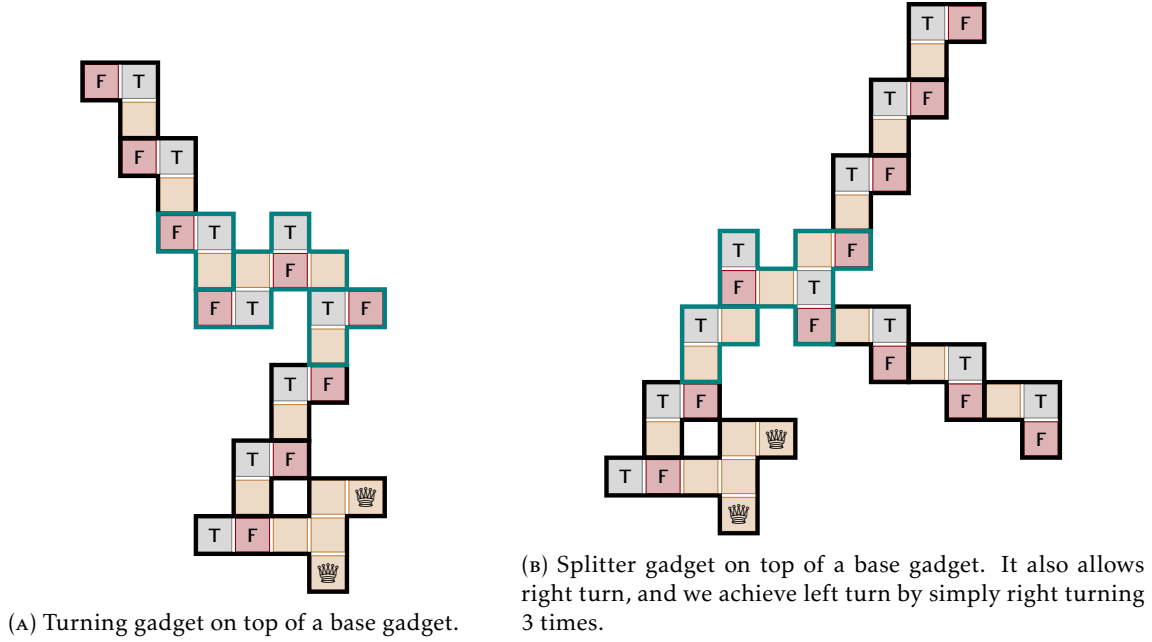

	\begin{lemma}\label{lem:GadgetSplitter}
		The gadget of Figure~\ref{fig:splitter-gadget} has two maximum queen placements of 13 queens different in the placement of 11 queens on the red `\textsf{F}' or gray `\textsf{T}' tiles. It doubles the variable $x$.
	\end{lemma}
	\begin{proof}
		The two maximum queen placements can be read of the tiles of Figure~\ref{fig:splitter-gadget}. Any number of extensions can be added without changing the fact that there will only be two solutions.
	\end{proof}

	Lemmata~\ref{lem:GadgetBase}--\ref{lem:GadgetSplitter} are meant to connect gadgets representing literal to clauses of an instance of $\psat$. Lemmata~\ref{lem:Gadget2Clause}--\ref{lem:Gadget3Clause} explain the clause gadgets. We show the correspondence of maximum queen placement with the realisation of a 2-clause for Lemma~\ref{lem:Gadget2Clause}, and leave out the 3-clause to the Appendix. (We recall that the applet~\cite{LRM26Software} contains all the gadgets as presets for testing.)
	
Before stating the lemma, we wish to mention that, since \psat{} is still NP-complete even when restricted to exactly 3-clauses~\cite[Lemma~1]{Mansfield83}, we could dispose of the 2-clause gadgets. They are, however, quite simple and elegant, and we prefer to reduce to the full \psat{}.
	
	\begin{figure}[htbp]
		\centering
		\begin{tikzpicture}[scale=.5,transform shape, baseline={(0,0)}]
		\foreach \x/\y in {
			3/9, 3/8, 4/8, 4/7, 5/7, 6/7, 3/6, 4/6, 3/5
		} {
			\path [draw=brown!, fill=brown!30] (.5+\x-0.45, .5+\y-1.45)
			-- ++(0,.9)
			-- ++(.9,0)
			-- ++(0,-.9)
			--cycle;
		}
		
		\foreach \x/\y in {6/7} {
			\node[anchor=center, scale=1.8] at (.5+\x, \y-.5) {\queen};
		}
		
		\draw[line width=1.6pt, black] (3,4) -- (4,4) -- (4,5) -- (5,5) -- (5,6) -- (3,6) -- cycle; 
		\draw[line width=1.6pt, black] (4,6) -- (7,6) -- (7,7) -- (4,7) -- cycle; 
		\draw[line width=1.6pt, black] (3,7) -- (5,7) -- (5,8) -- (4,8) -- (4,9) -- (3,9) -- cycle; 
		\end{tikzpicture}
		\qquad 
		\begin{tikzpicture}[scale=.5,transform shape]
		\foreach \x/\y in {
			4/13, 1/12, 2/12, 3/12, 4/12, 2/11, 4/11, 5/11, 2/10, 3/10, 3/9, 3/8, 4/8, 4/7, 5/7, 6/7, 3/6, 4/6, 3/5, 2/4, 3/4, 2/3, 4/3, 5/3, 1/2, 2/2, 3/2, 4/2, 4/1
		} {
			\path [draw=brown!, fill=brown!30] (.5+\x-0.45, .5+\y-1.45)
			-- ++(0,.9)
			-- ++(.9,0)
			-- ++(0,-.9)
			--cycle;
		}
		\foreach \x/\y in {
			1/2,1/12
		} {
			\path [draw=gray!, fill=gray!30] (.5+\x-0.45, .5+\y-1.45)
			-- ++(0,.9)
			-- ++(.9,0)
			-- ++(0,-.9)
			--cycle;
			\node[anchor=center, scale=1.5] at (.5+\x, \y-.5) {\textbf{\textsf{T}}};
		}
		
		\foreach \x/\y in {
			2/2,2/12
		} {
			\path [draw=sangria!, fill=sangria!30] (.5+\x-0.45, .5+\y-1.45)
			-- ++(0,.9)
			-- ++(.9,0)
			-- ++(0,-.9)
			--cycle;
			\node[anchor=center, scale=1.5] at (.5+\x, \y-.5) {\textbf{\textsf{F}}};
		}
		
		\foreach \x/\y in {4/13, 5/11, 6/7, 5/3, 4/1} {
			\node[anchor=center, scale=1.8] at (.5+\x, \y-.5) {\queen};
		}
		
		\draw[line width=1.6pt, black] (1,1) -- (4,1) -- (4,0) -- (5,0) -- (5,2) -- (6,2) -- (6,3) -- (4,3) -- (4,2) -- (1,2) -- cycle; 
		\draw[line width=1.6pt, black] (2,2) -- (3,2) -- (3,3) -- (4,3) -- (4,4) -- (2,4) -- cycle; 
		\draw[line width=1.6pt, black] (3,4) -- (4,4) -- (4,5) -- (5,5) -- (5,6) -- (3,6) -- cycle; 
		\draw[line width=1.6pt, black] (4,6) -- (7,6) -- (7,7) -- (4,7) -- cycle; 
		\draw[line width=1.6pt, black] (3,7) -- (5,7) -- (5,8) -- (4,8) -- (4,9) -- (3,9) -- cycle; 
		\draw[line width=1.6pt, black] (2,9) -- (4,9) -- (4,10) -- (3,10) -- (3,11) -- (2,11) -- cycle; 
		\draw[line width=1.6pt, black] (1,11) -- (4,11) -- (4,10) -- (6,10) -- (6,11) -- (5,11) -- (5,13) -- (4,13) -- (4,12) -- (1,12) -- cycle; 
		
		\node (lblA) at (-1.5, 2) {\textbf{\textsf{Base A}}};
		\draw[->, >=stealth, thick, line width=1.2pt] (lblA) to[out=0, in=180] (0.5, 1.5);
		
		\node (lblB) at (-1.5, 11) {\textbf{\textsf{Base B}}};
		\draw[->, >=stealth, thick, line width=1.2pt] (lblB) to[out=0, in=180] (0.5, 11.5);
		\end{tikzpicture}
		\caption{Left: 2-clause gadget. Right: a 2-clause gadget connected to two bases; there 11 queens can be placed unless $x_a = x_b = \texttt F$, so it encodes $x_a \lor x_b$.}
		\label{fig:clause-2-gadget}
	\end{figure}
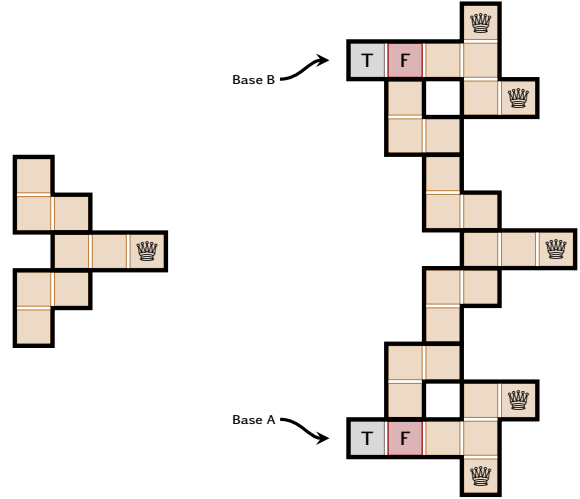
	
	\begin{lemma}\label{lem:Gadget2Clause}
		The gadget of Figure~\ref{fig:clause-2-gadget} represents a 2-clause $x_a\lor x_b$. It has 3 maximum queen placements corresponding to the three assignments of $(x_a,x_b)$ satisfying the clause $x_a\lor x_b$.
	\end{lemma}
	\begin{proof}
		We give the three maximum queen placements in Figure~\ref{table:queenplacement} below. There is no way to have both bases include a \texttt{F} placement (on red `\textsf{F}' tiles), and so the gadget can be seen as encoding the clause $x_a\lor x_b$.
		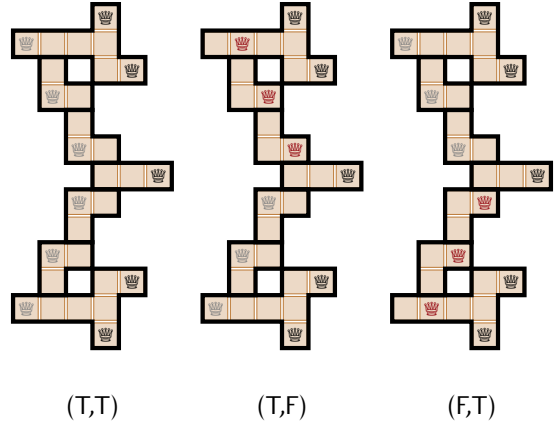
\begin{figure}[h]
			\centering
			\begin{tabular}{ccc}
				\begin{tikzpicture}[scale=.35,transform shape]
				\foreach \x/\y in {
					4/13, 1/12, 2/12, 3/12, 4/12, 2/11, 4/11, 5/11, 2/10, 3/10, 3/9, 3/8, 4/8, 4/7, 5/7, 6/7, 3/6, 4/6, 3/5, 2/4, 3/4, 2/3, 4/3, 5/3, 1/2, 2/2, 3/2, 4/2, 4/1
				} {
					\path [draw=brown!, fill=brown!30] (.5+\x-0.45, .5+\y-1.45)
					-- ++(0,.9)
					-- ++(.9,0)
					-- ++(0,-.9)
					--cycle;
				}
				
				\foreach \x/\y in {4/13, 5/11, 6/7, 5/3, 4/1} {
					\node[anchor=center, scale=1.8] at (.5+\x, \y-.5) {\queen};
				}
				\foreach \x/\y in {3/8,2/10,1/12, 1/2,2/4,3/6}
				{
					\node[anchor=center, scale=1.8, color=gray] at (.5+\x, \y-.5) {\queen};
				}
				\draw[line width=1.6pt, black] (1,1) -- (4,1) -- (4,0) -- (5,0) -- (5,2) -- (6,2) -- (6,3) -- (4,3) -- (4,2) -- (1,2) -- cycle; 
				\draw[line width=1.6pt, black] (2,2) -- (3,2) -- (3,3) -- (4,3) -- (4,4) -- (2,4) -- cycle; 
				\draw[line width=1.6pt, black] (3,4) -- (4,4) -- (4,5) -- (5,5) -- (5,6) -- (3,6) -- cycle; 
				\draw[line width=1.6pt, black] (4,6) -- (7,6) -- (7,7) -- (4,7) -- cycle; 
				\draw[line width=1.6pt, black] (3,7) -- (5,7) -- (5,8) -- (4,8) -- (4,9) -- (3,9) -- cycle; 
				\draw[line width=1.6pt, black] (2,9) -- (4,9) -- (4,10) -- (3,10) -- (3,11) -- (2,11) -- cycle; 
				\draw[line width=1.6pt, black] (1,11) -- (4,11) -- (4,10) -- (6,10) -- (6,11) -- (5,11) -- (5,13) -- (4,13) -- (4,12) -- (1,12) -- cycle; 
				\end{tikzpicture}
				&
				\begin{tikzpicture}[scale=.35,transform shape]
				\foreach \x/\y in {
					4/13, 1/12, 2/12, 3/12, 4/12, 2/11, 4/11, 5/11, 2/10, 3/10, 3/9, 3/8, 4/8, 4/7, 5/7, 6/7, 3/6, 4/6, 3/5, 2/4, 3/4, 2/3, 4/3, 5/3, 1/2, 2/2, 3/2, 4/2, 4/1
				} {
					\path [draw=brown!, fill=brown!30] (.5+\x-0.45, .5+\y-1.45)
					-- ++(0,.9)
					-- ++(.9,0)
					-- ++(0,-.9)
					--cycle;
				}
				
				\foreach \x/\y in {4/13, 5/11, 6/7, 5/3, 4/1} {
					\node[anchor=center, scale=1.8] at (.5+\x, \y-.5) {\queen};
				}
				\foreach \x/\y in {1/2,2/4,3/6}
				{
					\node[anchor=center, scale=1.8, color= gray] at (.5+\x, \y-.5) {\queen};
				}
				\foreach \x/\y in {4/8,3/10,2/12}
				{
					\node[anchor=center, scale=1.8, color=sangria] at (.5+\x, \y-.5) {\queen};
				}
				\draw[line width=1.6pt, black] (1,1) -- (4,1) -- (4,0) -- (5,0) -- (5,2) -- (6,2) -- (6,3) -- (4,3) -- (4,2) -- (1,2) -- cycle; 
				\draw[line width=1.6pt, black] (2,2) -- (3,2) -- (3,3) -- (4,3) -- (4,4) -- (2,4) -- cycle; 
				\draw[line width=1.6pt, black] (3,4) -- (4,4) -- (4,5) -- (5,5) -- (5,6) -- (3,6) -- cycle; 
				\draw[line width=1.6pt, black] (4,6) -- (7,6) -- (7,7) -- (4,7) -- cycle; 
				\draw[line width=1.6pt, black] (3,7) -- (5,7) -- (5,8) -- (4,8) -- (4,9) -- (3,9) -- cycle; 
				\draw[line width=1.6pt, black] (2,9) -- (4,9) -- (4,10) -- (3,10) -- (3,11) -- (2,11) -- cycle; 
				\draw[line width=1.6pt, black] (1,11) -- (4,11) -- (4,10) -- (6,10) -- (6,11) -- (5,11) -- (5,13) -- (4,13) -- (4,12) -- (1,12) -- cycle; 
				\end{tikzpicture}
				&
				\begin{tikzpicture}[scale=.35,transform shape]
				\foreach \x/\y in {
					4/13, 1/12, 2/12, 3/12, 4/12, 2/11, 4/11, 5/11, 2/10, 3/10, 3/9, 3/8, 4/8, 4/7, 5/7, 6/7, 3/6, 4/6, 3/5, 2/4, 3/4, 2/3, 4/3, 5/3, 1/2, 2/2, 3/2, 4/2, 4/1
				} {
					\path [draw=brown!, fill=brown!30] (.5+\x-0.45, .5+\y-1.45)
					-- ++(0,.9)
					-- ++(.9,0)
					-- ++(0,-.9)
					--cycle;
				}
				
				\foreach \x/\y in {4/13, 5/11, 6/7, 5/3, 4/1} {
					\node[anchor=center, scale=1.8] at (.5+\x, \y-.5) {\queen};
				}
				\foreach \x/\y in {2/2,3/4,4/6}
				{
					\node[anchor=center, scale=1.8, color= sangria] at (.5+\x, \y-.5) {\queen};
				}
				\foreach \x/\y in {3/8,2/10,1/12}
				{
					\node[anchor=center, scale=1.8, color=gray] at (.5+\x, \y-.5) {\queen};
				}
				
				\draw[line width=1.6pt, black] (1,1) -- (4,1) -- (4,0) -- (5,0) -- (5,2) -- (6,2) -- (6,3) -- (4,3) -- (4,2) -- (1,2) -- cycle; 
				\draw[line width=1.6pt, black] (2,2) -- (3,2) -- (3,3) -- (4,3) -- (4,4) -- (2,4) -- cycle; 
				\draw[line width=1.6pt, black] (3,4) -- (4,4) -- (4,5) -- (5,5) -- (5,6) -- (3,6) -- cycle; 
				\draw[line width=1.6pt, black] (4,6) -- (7,6) -- (7,7) -- (4,7) -- cycle; 
				\draw[line width=1.6pt, black] (3,7) -- (5,7) -- (5,8) -- (4,8) -- (4,9) -- (3,9) -- cycle; 
				\draw[line width=1.6pt, black] (2,9) -- (4,9) -- (4,10) -- (3,10) -- (3,11) -- (2,11) -- cycle; 
				\draw[line width=1.6pt, black] (1,11) -- (4,11) -- (4,10) -- (6,10) -- (6,11) -- (5,11) -- (5,13) -- (4,13) -- (4,12) -- (1,12) -- cycle; 
				\end{tikzpicture}\\
				&&\\
				(\textsf{T,T}) & 	(\textsf{T,F}) & 	(\textsf{F,T})
			\end{tabular}
			\caption{Three maximum queen placements for the 2-clause $x_a\lor x_b$ and their corresponding literals below.}\label{table:queenplacement}
		\end{figure}		
		Note that the connection to the variable gadgets can be made longer by adding extension gadgets, or any of the turner, shifting or splitting gadgets.
	\end{proof}

	\begin{lemma}
		\label{lem:Gadget3Clause}
		The gadget of Figure~\ref{fig:clause-3-gadget} represents a 3-clause $x_a\lor x_b\lor x_c$. It has 7 maximum queen placements corresponding 7 assignments of $(x_a,x_b,x_c)$ satisfying $x_a\lor x_b\lor x_c$
	\end{lemma}
	\begin{proof}
		There are 7 maximum placements of the clause gadgets and they correspond to the 7 solutions of the clause $x_a\lor x_b \lor x_d$; see Figure~\ref{table:3clausequeenplacement}.
	\end{proof}

	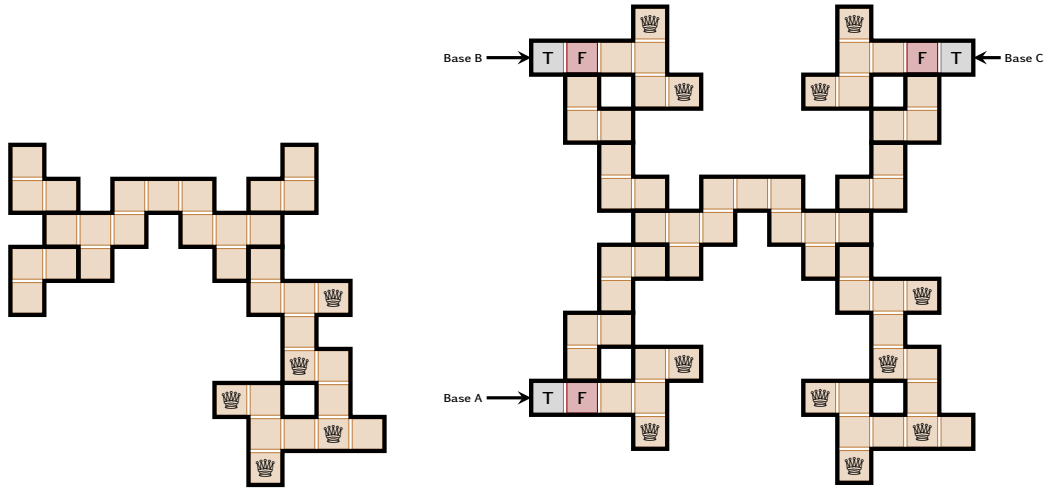
\begin{figure}[htbp]
		\centering
		\begin{tikzpicture}[scale=.45,transform shape,baseline={(0,0)}]
		\foreach \x/\y in {
			3/10, 11/10,
			3/9, 4/9, 6/9, 7/9, 8/9, 10/9, 11/9,
			4/8, 5/8, 6/8, 8/8, 9/8, 10/8,
			3/7, 4/7, 5/7, 9/7, 10/7,
			3/6, 10/6, 11/6, 12/6,
			11/5,
			11/4, 12/4,
			9/3, 10/3, 12/3,
			10/2, 11/2, 12/2, 13/2,
			10/1
		} {
			\path [draw=brown!, fill=brown!30] (.5+\x-0.45, .5+\y-1.45)
			-- ++(0,.9)
			-- ++(.9,0)
			-- ++(0,-.9)
			--cycle;
		}
		
		\foreach \x/\y in { 12/6,  11/4, 9/3, 12/2, 10/1} {
			\node[anchor=center, scale=1.8] at (.5+\x, \y-.5) {\queen};
		}
		
		\draw[line width=1.6pt, black] (10,0) -- (11,0) -- (11,1) -- (14,1) -- (14,2) -- (13,2) -- (13,4) -- (12,4) -- (12,5) -- (13,5) -- (13,6) -- (11,6) -- (11,7) -- (10,7) -- (10,5) -- (11,5) -- (11,3) -- (12,3) -- (12,2) -- (11,2) -- (11,3) -- (9,3) -- (9,2) -- (10,2) -- (10,0) -- cycle;
		\draw[line width=1.6pt, black] (3,5) -- (4,5) -- (4,6) -- (5,6) -- (5,7) -- (3,7) -- cycle; 
		\draw[line width=1.6pt, black] (3,8) -- (5,8) -- (5,9) -- (4,9) -- (4,10) -- (3,10) -- cycle; 
		\draw[line width=1.6pt, black] (12,8) -- (10,8) -- (10,9) -- (11,9) -- (11,10) -- (12,10) -- cycle; 
		\draw[line width=1.6pt, black] (4,7) -- (5,7) -- (5,6) -- (6,6) -- (6,7) -- (7,7) -- (7,8) -- (8,8) -- (8,7) -- (9,7) -- (9,6) -- (10,6) -- (10,7) -- (11,7) -- (11,8) -- (10,8) -- (9,8) -- (9,9) -- (6,9) -- (6,8) -- (4,8) -- cycle;
		\end{tikzpicture}\qquad
		\begin{tikzpicture}[scale=.45,transform shape]
		\foreach \x/\y in {
			4/14, 10/14,
			3/13, 4/13, 10/13, 11/13, 
			2/12, 4/12, 5/12, 9/12, 10/12, 12/12,
			2/11, 3/11, 11/11, 12/11,
			3/10, 11/10,
			3/9, 4/9, 6/9, 7/9, 8/9, 10/9, 11/9,
			4/8, 5/8, 6/8, 8/8, 9/8, 10/8,
			3/7, 4/7, 5/7, 9/7, 10/7,
			3/6, 10/6, 11/6, 12/6,
			2/5, 3/5, 11/5,
			2/4, 4/4, 5/4, 11/4, 12/4,
			3/3, 4/3, 9/3, 10/3, 12/3,
			4/2, 10/2, 11/2, 12/2, 13/2,
			10/1
		} {
			\path [draw=brown!, fill=brown!30] (.5+\x-0.45, .5+\y-1.45)
			-- ++(0,.9)
			-- ++(.9,0)
			-- ++(0,-.9)
			--cycle;
		}
		
		\foreach \x/\y in {
			1/3,1/13,13/13
		} {
			\path [draw=gray!, fill=gray!30] (.5+\x-0.45, .5+\y-1.45)
			-- ++(0,.9)
			-- ++(.9,0)
			-- ++(0,-.9)
			--cycle;
			\node[anchor=center, scale=1.5] at (.5+\x, \y-.5) {\textbf{\textsf{T}}};
		}
		
		\foreach \x/\y in {
			2/3,2/13,12/13
		} {
			\path [draw=sangria!, fill=sangria!30] (.5+\x-0.45, .5+\y-1.45)
			-- ++(0,.9)
			-- ++(.9,0)
			-- ++(0,-.9)
			--cycle;
			\node[anchor=center, scale=1.5] at (.5+\x, \y-.5) {\textbf{\textsf{F}}};
		}
		
		\foreach \x/\y in {4/14, 10/14, 5/12, 9/12, 12/6, 5/4, 11/4, 9/3, 4/2, 12/2, 10/1} {
			\node[anchor=center, scale=1.8] at (.5+\x, \y-.5) {\queen};
		}
		
		\draw[line width=1.6pt, black] (1,2) -- (4,2) -- (4,1) -- (5,1) -- (5,3) -- (6,3) -- (6,4) -- (4,4) -- (4,3) -- (1,3) -- cycle; 
		\draw[line width=1.6pt, black] (1,12) -- (4,12) -- (4,11) -- (6,11) -- (6,12) -- (5,12) -- (5,14) -- (4,14) -- (4,13) -- (1,13) -- cycle; 
		
		\draw[line width=1.6pt, black] (10,0) -- (11,0) -- (11,1) -- (14,1) -- (14,2) -- (13,2) -- (13,4) -- (12,4) -- (12,5) -- (13,5) -- (13,6) -- (11,6) -- (11,7) -- (10,7) -- (10,5) -- (11,5) -- (11,3) -- (12,3) -- (12,2) -- (11,2) -- (11,3) -- (9,3) -- (9,2) -- (10,2) -- (10,0) -- cycle;
		
		\draw[line width=1.6pt, black] (14,12) -- (11,12) -- (11,11) -- (9,11) -- (9,12) -- (10,12) -- (10,14) -- (11,14) -- (11,13) -- (14,13) -- cycle; 
		\draw[line width=1.6pt, black] (2,3) -- (3,3) -- (3,4) -- (4,4) -- (4,5) -- (2,5) -- cycle; 
		\draw[line width=1.6pt, black] (3,5) -- (4,5) -- (4,6) -- (5,6) -- (5,7) -- (3,7) -- cycle; 
		\draw[line width=1.6pt, black] (2,10) -- (4,10) -- (4,11) -- (3,11) -- (3,12) -- (2,12) -- cycle; 
		\draw[line width=1.6pt, black] (3,8) -- (5,8) -- (5,9) -- (4,9) -- (4,10) -- (3,10) -- cycle; 
		\draw[line width=1.6pt, black] (13,10) -- (11,10) -- (11,11) -- (12,11) -- (12,12) -- (13,12) -- cycle; 
		\draw[line width=1.6pt, black] (12,8) -- (10,8) -- (10,9) -- (11,9) -- (11,10) -- (12,10) -- cycle; 
		
		\draw[line width=1.6pt, black] (4,7) -- (5,7) -- (5,6) -- (6,6) -- (6,7) -- (7,7) -- (7,8) -- (8,8) -- (8,7) -- (9,7) -- (9,6) -- (10,6) -- (10,7) -- (11,7) -- (11,8) -- (10,8) -- (9,8) -- (9,9) -- (6,9) -- (6,8) -- (4,8) -- cycle;
		
		\node (lblA) at (-1, 2.5) {\textbf{\textsf{Base A}}};
		\draw[->, >=stealth, thick, line width=1.2pt] (lblA) to[out=0, in=180] (1, 2.5);
		
		\node (lblB) at (-1, 12.5) {\textbf{\textsf{Base B}}};
		\draw[->, >=stealth, thick, line width=1.2pt] (lblB) to[out=0, in=180] (1, 12.5);
		
		\node (lblD) at (15.5, 12.5) {\textbf{\textsf{Base C}}};
		\draw[->, >=stealth, thick, line width=1.2pt] (lblD) to[out=180, in=0] (14, 12.5);
		\end{tikzpicture}
		\caption{Left: the 3-clause gadget. Right: the 3-clause connected to three variable gadgets; 23 queens can be placed unless $x_a = x_b = x_c = \texttt F$, so we encode $x_a \lor x_b \lor x_c$)}
		\label{fig:clause-3-gadget}
	\end{figure}

	We now introduce a final set of gadgets: the shifting gadgets. They allow us to adjust the position of the extension and help connect the gadgets together.
	
	\begin{lemma}\label{lem:shiftinggadget}
		The three gadgets of Figures~\ref{fig:shiftingleftright}--\ref{fig:shiftingforward} admit only two maximum placements of queens, and thus allow to shift the signal by one in the left, right and forward directions, respectively.
	\end{lemma}
	
	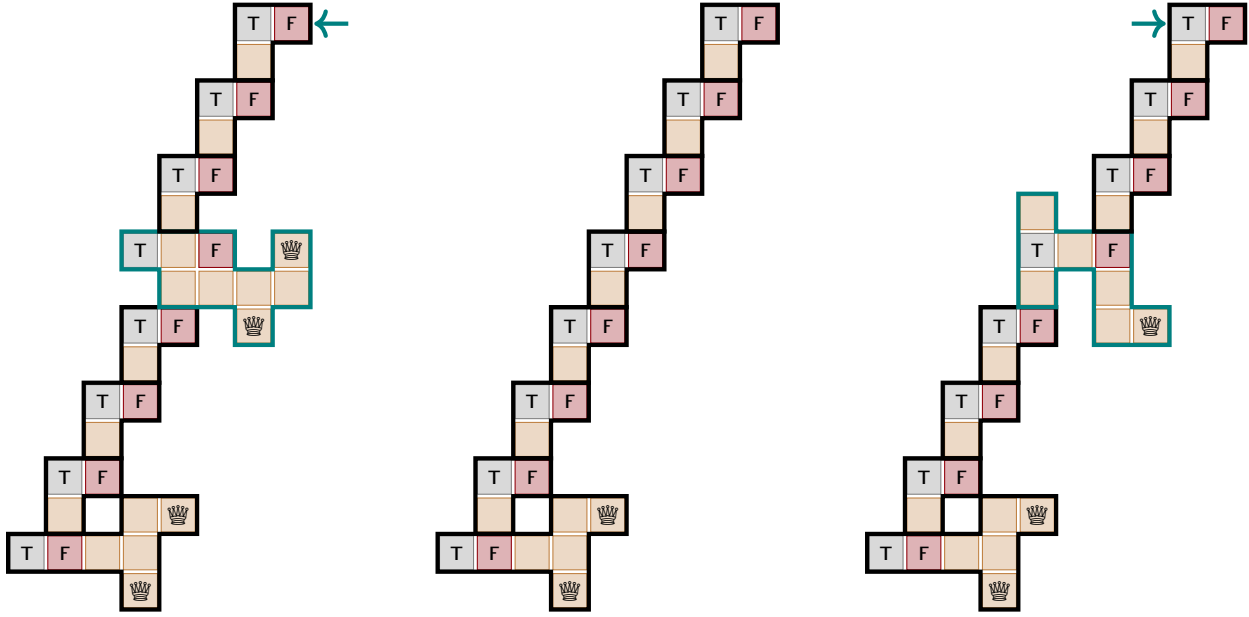
\begin{figure}[h]
		\centering
		\hfill
		
		\begin{tikzpicture}[scale=.5,transform shape]
		\foreach \x/\y in {
			4/7, 3/5, 2/3, 4/3, 5/3, 3/2, 4/2, 4/1,
			7/8,
			5/9,6/9,7/9,8/9,
			5/10,6/10,8/10,
			5/11,6/13,7/15
		} {
			\path [draw=brown!, fill=brown!30] (.5+\x-0.45, .5+\y-1.45)
			-- ++(0,.9)
			-- ++(.9,0)
			-- ++(0,-.9)
			--cycle;
		}
		
		\foreach \x/\y in {
			4/8, 3/6, 2/4, 1/2,
			4/10,5/12,6/14,7/16
		} {
			\path [draw=gray!, fill=gray!30] (.5+\x-0.45, .5+\y-1.45)
			-- ++(0,.9)
			-- ++(.9,0)
			-- ++(0,-.9)
			--cycle;
			\node[anchor=center, scale=1.5] at (.5+\x, \y-.5) {\textbf{\textsf{T}}};
		}
		
		\foreach \x/\y in {
			5/8, 4/6, 3/4, 2/2,
			6/10,6/12,7/14,8/16
		} {
			\path [draw=sangria!, fill=sangria!30] (.5+\x-0.45, .5+\y-1.45)
			-- ++(0,.9)
			-- ++(.9,0)
			-- ++(0,-.9)
			--cycle;
			\node[anchor=center, scale=1.5] at (.5+\x, \y-.5) {\textbf{\textsf{F}}};
		}
		
		\foreach \x/\y in {5/3, 4/1,7/8,8/10} {
			\node[anchor=center, scale=1.8] at (.5+\x, \y-.5) {\queen};
		}
		
		\draw[line width=1.6pt, black] (1,1) -- (4,1) -- (4,0) -- (5,0) -- (5,2) -- (6,2) -- (6,3) -- (4,3) -- (4,2) -- (1,2) -- cycle; 
		\draw[line width=1.6pt, black] (2,2) -- (3,2) -- (3,3) -- (4,3) -- (4,4) -- (2,4) -- cycle; 
		\draw[line width=1.6pt, black] (3,4) -- (4,4) -- (4,5) -- (5,5) -- (5,6) -- (3,6) -- cycle; 
		\draw[line width=1.6pt, black] (4,6) -- (5,6) -- (5,7) -- (6,7) -- (6,8) -- (4,8) -- cycle; 
		\draw[line width=1.6pt, teal] (5,8) -- (7,8)-- (7,7) -- (8,7) --(8,8) -- (9,8) -- (9,10) --(8,10) -- (8,9) -- (7,9)--(7,10) -- (4,10)--(4,9) -- (5,9) --cycle; 
		\draw[line width=1.6pt, black] (2+3,2+8) -- (3+3,2+8) -- (3+3,3+8) -- (4+3,3+8) -- (4+3,4+8) -- (2+3,4+8) -- cycle; 
		\draw[line width=1.6pt, black] (3+3,4+8) -- (4+3,4+8) -- (4+3,5+8) -- (5+3,5+8) -- (5+3,6+8) -- (3+3,6+8) -- cycle; 
		\draw[line width=1.6pt, black] (4+3,6+8) -- (5+3,6+8) -- (5+3,7+8) -- (6+3,7+8) -- (6+3,8+8) -- (4+3,8+8) -- cycle; 
		\draw[->, ultra thick,teal] (10,15.5)--(9.1,15.5);
		\end{tikzpicture}
		\hfill
		\begin{tikzpicture}[scale=.5,transform shape]
		\foreach \x/\y in {
			4/7, 3/5, 2/3, 4/3, 5/3, 3/2, 4/2, 4/1,
			5/9, 6/11,7/13,8/15
		} {
			\path [draw=brown!, fill=brown!30] (.5+\x-0.45, .5+\y-1.45)
			-- ++(0,.9)
			-- ++(.9,0)
			-- ++(0,-.9)
			--cycle;
		}
		
		\foreach \x/\y in {
			4/8, 3/6, 2/4, 1/2,
			5/10,6/12,7/14,8/16
		} {
			\path [draw=gray!, fill=gray!30] (.5+\x-0.45, .5+\y-1.45)
			-- ++(0,.9)
			-- ++(.9,0)
			-- ++(0,-.9)
			--cycle;
			\node[anchor=center, scale=1.5] at (.5+\x, \y-.5) {\textbf{\textsf{T}}};
		}
		
		\foreach \x/\y in {
			5/8, 4/6, 3/4, 2/2,
			6/10,7/12,8/14,9/16
		} {
			\path [draw=sangria!, fill=sangria!30] (.5+\x-0.45, .5+\y-1.45)
			-- ++(0,.9)
			-- ++(.9,0)
			-- ++(0,-.9)
			--cycle;
			\node[anchor=center, scale=1.5] at (.5+\x, \y-.5) {\textbf{\textsf{F}}};
		}
		
		\foreach \x/\y in {5/3, 4/1} {
			\node[anchor=center, scale=1.8] at (.5+\x, \y-.5) {\queen};
		}
		
		\draw[line width=1.6pt, black] (1,1) -- (4,1) -- (4,0) -- (5,0) -- (5,2) -- (6,2) -- (6,3) -- (4,3) -- (4,2) -- (1,2) -- cycle; 
		\draw[line width=1.6pt, black] (2,2) -- (3,2) -- (3,3) -- (4,3) -- (4,4) -- (2,4) -- cycle; 
		\draw[line width=1.6pt, black] (3,4) -- (4,4) -- (4,5) -- (5,5) -- (5,6) -- (3,6) -- cycle; 
		\draw[line width=1.6pt, black] (4,6) -- (5,6) -- (5,7) -- (6,7) -- (6,8) -- (4,8) -- cycle; 
		\draw[line width=1.6pt, black] (2+3,2+6) -- (3+3,2+6) -- (3+3,3+6) -- (4+3,3+6) -- (4+3,4+6) -- (2+3,4+6) -- cycle; 
		\draw[line width=1.6pt, black] (2+4,2+8) -- (3+4,2+8) -- (3+4,3+8) -- (4+4,3+8) -- (4+4,4+8) -- (2+4,4+8) -- cycle; 
		\draw[line width=1.6pt, black] (3+4,4+8) -- (4+4,4+8) -- (4+4,5+8) -- (5+4,5+8) -- (5+4,6+8) -- (3+4,6+8) -- cycle; 
		\draw[line width=1.6pt, black] (4+4,6+8) -- (5+4,6+8) -- (5+4,7+8) -- (6+4,7+8) -- (6+4,8+8) -- (4+4,8+8) -- cycle; 
		\end{tikzpicture}
		\hfill
		\begin{tikzpicture}[scale=.5,transform shape]
		\foreach \x/\y in {
			4/7, 3/5, 2/3, 4/3, 5/3, 3/2, 4/2, 4/1,
			7/8,8/8,
			5/9,7/9,
			5/10,6/10,7/10,
			5/11,7/11,
			8/13,
			9/15
		} {
			\path [draw=brown!, fill=brown!30] (.5+\x-0.45, .5+\y-1.45)
			-- ++(0,.9)
			-- ++(.9,0)
			-- ++(0,-.9)
			--cycle;
		}
		
		\foreach \x/\y in {
			4/8, 3/6, 2/4, 1/2,
			5/10,
			7/12,
			8/14,
			9/16
		} {
			\path [draw=gray!, fill=gray!30] (.5+\x-0.45, .5+\y-1.45)
			-- ++(0,.9)
			-- ++(.9,0)
			-- ++(0,-.9)
			--cycle;
			\node[anchor=center, scale=1.5] at (.5+\x, \y-.5) {\textbf{\textsf{T}}};
		}
		
		\foreach \x/\y in {
			5/8, 4/6, 3/4, 2/2,
			7/10,
			8/12,
			9/14,
			10/16
		} {
			\path [draw=sangria!, fill=sangria!30] (.5+\x-0.45, .5+\y-1.45)
			-- ++(0,.9)
			-- ++(.9,0)
			-- ++(0,-.9)
			--cycle;
			\node[anchor=center, scale=1.5] at (.5+\x, \y-.5) {\textbf{\textsf{F}}};
		}
		
		\foreach \x/\y in {5/3, 4/1,8/8} {
			\node[anchor=center, scale=1.8] at (.5+\x, \y-.5) {\queen};
		}
		
		\draw[line width=1.6pt, black] (1,1) -- (4,1) -- (4,0) -- (5,0) -- (5,2) -- (6,2) -- (6,3) -- (4,3) -- (4,2) -- (1,2) -- cycle; 
		\draw[line width=1.6pt, black] (2,2) -- (3,2) -- (3,3) -- (4,3) -- (4,4) -- (2,4) -- cycle; 
		\draw[line width=1.6pt, black] (3,4) -- (4,4) -- (4,5) -- (5,5) -- (5,6) -- (3,6) -- cycle; 
		\draw[line width=1.6pt, black] (4,6) -- (5,6) -- (5,7) -- (6,7) -- (6,8) -- (4,8) -- cycle; 
		\draw[line width=1.6pt, teal] (6,8) -- (6,9) -- (7,9) -- (7,7) --(9,7) -- (9,8) -- (8,8) --(8,10) -- (6,10) -- (6,11)--(5,11) -- (5,8) --cycle; 
		\draw[line width=1.6pt, black] (2+5,2+8) -- (3+5,2+8) -- (3+5,3+8) -- (4+5,3+8) -- (4+5,4+8) -- (2+5,4+8) -- cycle; 
		\draw[line width=1.6pt, black] (3+5,4+8) -- (4+5,4+8) -- (4+5,5+8) -- (5+5,5+8) -- (5+5,6+8) -- (3+5,6+8) -- cycle; 
		\draw[line width=1.6pt, black] (4+5,6+8) -- (5+5,6+8) -- (5+5,7+8) -- (6+5,7+8) -- (6+5,8+8) -- (4+5,8+8) -- cycle; 
		\draw[->, ultra thick,teal] (8,15.5)--(9-.1,15.5);
		\end{tikzpicture}
		\hfill
		\caption{The left and right shifting gadgets, respectively left and right to the normal base gadget with extensions. They allow adjusting the position of the last extension by 1 tile right or left.}\label{fig:shiftingleftright}
	\end{figure}
	
	\begin{figure}
		\begin{tikzpicture}[scale=.5,
		rotate=270,every node/.style={scale=.5} 
		] 
		\foreach \x/\y in {
			4/7, 3/5, 2/3, 4/3, 5/3, 3/2, 4/2, 4/1,
			5/9, 7/9,
			6/10,7/10,
			6/11,7/11,
			6/12,7/12,8/12,
			7/13,8/13,
			8/14,
			9/16,
			12/17,
			10/18,11/18,12/18,13/18,
			10/19,11/19,13/19,
			10/20,11/22,12/24
		} {
			\path [draw=brown!, fill=brown!30] (.5+\x-0.45, .5+\y-1.45)
			-- ++(0,.9)
			-- ++(.9,0)
			-- ++(0,-.9)
			--cycle;
		}
		
		\foreach \x/\y in {
			4/8, 3/6, 2/4, 1/2,
			5/10,
			6/12,
			8/15,
			9/17,
			9/19,
			10/21,11/23,12/25
		} {
			\path [draw=gray!, fill=gray!30] (.5+\x-0.45, .5+\y-1.45)
			-- ++(0,.9)
			-- ++(.9,0)
			-- ++(0,-.9)
			--cycle;
			\node[anchor=center, scale=1.5] at (.5+\x, \y-.5) {\textbf{\textsf{T}}};
		}
		
		\foreach \x/\y in {
			5/8, 4/6, 3/4, 2/2,
			6/11,
			8/12,
			9/15,
			10/17,
			11/19,
			11/21,12/23,13/25
		} {
			\path [draw=sangria!, fill=sangria!30] (.5+\x-0.45, .5+\y-1.45)
			-- ++(0,.9)
			-- ++(.9,0)
			-- ++(0,-.9)
			--cycle;
			\node[anchor=center, scale=1.5] at (.5+\x, \y-.5) {\textbf{\textsf{F}}};
		}
		
		\foreach \x/\y in {5/3, 4/1,7/9,12/17,13/19} {
			\node[anchor=center, scale=1.8] at (.5+\x, \y-.5) {\queen};
		}
		
		\draw[line width=1.6pt, black] (1,1) -- (4,1) -- (4,0) -- (5,0) -- (5,2) -- (6,2) -- (6,3) -- (4,3) -- (4,2) -- (1,2) -- cycle; 
		\draw[line width=1.6pt, black] (2,2) -- (3,2) -- (3,3) -- (4,3) -- (4,4) -- (2,4) -- cycle; 
		\draw[line width=1.6pt, black] (3,4) -- (4,4) -- (4,5) -- (5,5) -- (5,6) -- (3,6) -- cycle; 
		\draw[line width=1.6pt, black] (4,6) -- (5,6) -- (5,7) -- (6,7) -- (6,8) -- (4,8) -- cycle; 
		\draw[line width=1.6pt, teal] (5,8) -- (6,8)-- (6,9) -- (7,9) --(7,8) -- (8,8) -- (8,11)--(9,11) -- (9,13) -- (7,13)--(7,12) -- (6,12)--(6,10) -- (5,10) --cycle; 
		\draw[line width=1.6pt, black] (2+6,2+11) -- (3+6,2+11) -- (3+6,3+11) -- (4+6,3+11) -- (4+6,4+11) -- (2+6,4+11) -- cycle; 
		\draw[line width=1.6pt, black] (2+7,2+13) -- (3+7,2+13) -- (3+7,3+13) -- (4+7,3+13) -- (4+7,4+13) -- (2+7,4+13) -- cycle; 
		\draw[line width=1.6pt, teal] (5+5,8+9) -- (7+5,8+9)-- (7+5,7+9) -- (8+5,7+9) --(8+5,8+9) -- (9+5,8+9) -- (9+5,10+9) --(8+5,10+9) -- (8+5,9+9) -- (7+5,9+9)--(7+5,10+9) -- (4+5,10+9)--(4+5,9+9) -- (5+5,9+9) --cycle; 
		\draw[line width=1.6pt, black] (2+8,2+17) -- (3+8,2+17) -- (3+8,3+17) -- (4+8,3+17) -- (4+8,4+17) -- (2+8,4+17) -- cycle; 
		\draw[line width=1.6pt, black] (2+9,2+19) -- (3+9,2+19) -- (3+9,3+19) -- (4+9,3+19) -- (4+9,4+19) -- (2+9,4+19) -- cycle; 
		\draw[line width=1.6pt, black] (2+10,2+21) -- (3+10,2+21) -- (3+10,3+21) -- (4+10,3+21) -- (4+10,4+21) -- (2+10,4+21) -- cycle; 
		\draw[->, ultra thick,teal] (11,24)--(11,25);
		\begin{scope}[shift = {(8,0)}]
		\foreach \x/\y in {
			4/7, 3/5, 2/3, 4/3, 5/3, 3/2, 4/2, 4/1,
			5/9, 6/11,7/13,8/15,
			9/17,10/19,11/21,12/23
		} {
			\path [draw=brown!, fill=brown!30] (.5+\x-0.45, .5+\y-1.45)
			-- ++(0,.9)
			-- ++(.9,0)
			-- ++(0,-.9)
			--cycle;
		}
		
		\foreach \x/\y in {
			4/8, 3/6, 2/4, 1/2,
			5/10,6/12,7/14,8/16,
			9/18,10/20,11/22,12/24
		} {
			\path [draw=gray!, fill=gray!30] (.5+\x-0.45, .5+\y-1.45)
			-- ++(0,.9)
			-- ++(.9,0)
			-- ++(0,-.9)
			--cycle;
			\node[anchor=center, scale=1.5] at (.5+\x, \y-.5) {\textbf{\textsf{T}}};
		}
		
		\foreach \x/\y in {
			5/8, 4/6, 3/4, 2/2,
			6/10,7/12,8/14,9/16,
			10/18,11/20,12/22,13/24
		} {
			\path [draw=sangria!, fill=sangria!30] (.5+\x-0.45, .5+\y-1.45)
			-- ++(0,.9)
			-- ++(.9,0)
			-- ++(0,-.9)
			--cycle;
			\node[anchor=center, scale=1.5] at (.5+\x, \y-.5) {\textbf{\textsf{F}}};
		}
		
		\foreach \x/\y in {5/3, 4/1} {
			\node[anchor=center, scale=1.8] at (.5+\x, \y-.5) {\queen};
		}
		
		\draw[line width=1.6pt, black] (1,1) -- (4,1) -- (4,0) -- (5,0) -- (5,2) -- (6,2) -- (6,3) -- (4,3) -- (4,2) -- (1,2) -- cycle; 
		\draw[line width=1.6pt, black] (2,2) -- (3,2) -- (3,3) -- (4,3) -- (4,4) -- (2,4) -- cycle; 
		\draw[line width=1.6pt, black] (3,4) -- (4,4) -- (4,5) -- (5,5) -- (5,6) -- (3,6) -- cycle; 
		\draw[line width=1.6pt, black] (4,6) -- (5,6) -- (5,7) -- (6,7) -- (6,8) -- (4,8) -- cycle; 
		\draw[line width=1.6pt, black] (2+3,2+6) -- (3+3,2+6) -- (3+3,3+6) -- (4+3,3+6) -- (4+3,4+6) -- (2+3,4+6) -- cycle; 
		\draw[line width=1.6pt, black] (2+4,2+8) -- (3+4,2+8) -- (3+4,3+8) -- (4+4,3+8) -- (4+4,4+8) -- (2+4,4+8) -- cycle; 
		\draw[line width=1.6pt, black] (3+4,4+8) -- (4+4,4+8) -- (4+4,5+8) -- (5+4,5+8) -- (5+4,6+8) -- (3+4,6+8) -- cycle; 
		\draw[line width=1.6pt, black] (4+4,6+8) -- (5+4,6+8) -- (5+4,7+8) -- (6+4,7+8) -- (6+4,8+8) -- (4+4,8+8) -- cycle; 
		\draw[line width=1.6pt, black] (4+5,6+10) -- (5+5,6+10) -- (5+5,7+10) -- (6+5,7+10) -- (6+5,8+10) -- (4+5,8+10) -- cycle; 
		\draw[line width=1.6pt, black] (4+6,6+12) -- (5+6,6+12) -- (5+6,7+12) -- (6+6,7+12) -- (6+6,8+12) -- (4+6,8+12) -- cycle; 
		\draw[line width=1.6pt, black] (4+7,6+14) -- (5+7,6+14) -- (5+7,7+14) -- (6+7,7+14) -- (6+7,8+14) -- (4+7,8+14) -- cycle; 
		\draw[line width=1.6pt, black] (4+8,6+16) -- (5+8,6+16) -- (5+8,7+16) -- (6+8,7+16) -- (6+8,8+16) -- (4+8,8+16) -- cycle; 
		\end{scope}
		\end{tikzpicture}
		\caption{Forward gadget with comparison below. It adjusts the position forward by 1 tile. }\label{fig:shiftingforward}
	\end{figure}
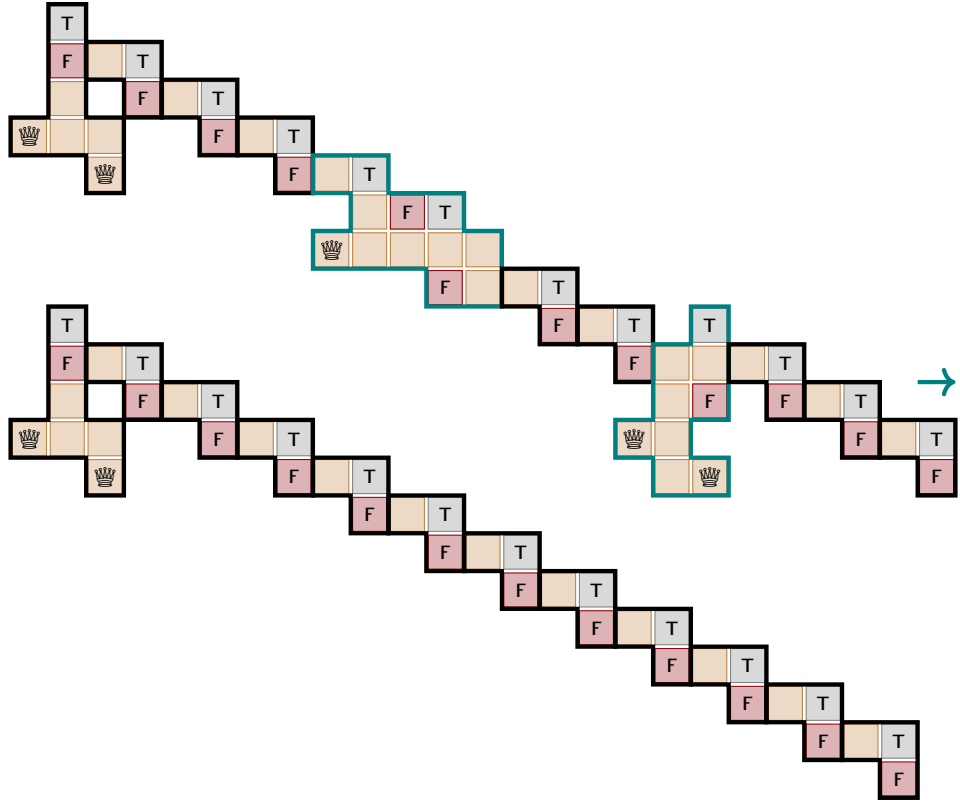

	The following proposition is a crucial step showing the translation between a \psat{} instance and an \indQueens{} instance, and showing that the polyomino thus constructed is done in polynomial time and that its size is bounded polyominially.

	\begin{proposition}\label{prop:InstanceSizePoly}
		Given an instance $C$ of \psat, it is possible to construct a polyomino $P(C)$ and a corresponding number of queens $n_{C}$ in polynomial time such that no more than $n_{C}$ independent queens can be placed on $P(C)$, and that $n_{C}$ independent queens can only be placed if $C$ is satisfiable.
	\end{proposition}
	\begin{proof}
		Let $C$ be an instance of \psat. As $C$ is planar, we use Proposition~\ref{prop:gridpath} to encode it into a grid path; this translation is made in polynomial time and is of polynomial size. We then enlarge each square of the grid path representation to be sufficiently large to welcome the gadgets presented in Lemmata~\ref{lem:GadgetBase}--\ref{lem:Gadget3Clause}.  
		Each of these can be done by enlarging all boxes of the grid path by a common factor, thus preserving the polynomial size. To make sure that the boxes connect correctly, we use the extension, the turning and the shifting gadgets to reach each side of a single box in the same way.
		
		If a variable is negated in a clause, we add an inverter gadget before it. Once all variable gadgets are connected to the related clauses, the polyomino $P(C)$ is constructed.
		
		We obtain $n_{C}$ by simply adding the number of queens for each gadget we used. Then, by the design of the gadgets, no more than $n_{C}$ non-attacking queens can be placed on $P(C)$. And if one wants to place $n_{C}$ non-attacking queens, the only possibilities that remain are those from the gadget options, which enforces that each clause gadget is connected to at least one variable gadget set to true.
	\end{proof}
	
	We are now ready to prove Theorem~\ref{mainres:NPQueenPolyo}.
	
	\begin{proof}[Proof of Theorem~\ref{mainres:NPQueenPolyo}]
		We have shown in Lemma~\ref{lem:Verif} that the verification is done in polynomial time. Then, using Proposition~\ref{prop:InstanceSizePoly}, we can create a polyomino $P(C)$ of polynomial size along with a number of queens $n_{C}$ from an instance $C$ of \psat.
		
		If there were a polynomial time algorithm to decide whether $n_{C}$ independent queens can be placed on $P(C)$, we could decide whether $C$ is satisfiable. This proves the NP-hardness of the problem, and thus we have finished the proof of its NP-completeness.
	\end{proof}
	
	Before closing the section, we state a quick corollary of Proposition~\ref{prop:InstanceSizePoly}.
	\begin{corollary}\label{coro:parsimonious}
		Let $C$ be an instance of \psat. The number of maximum queen placements in polyomino $P(C)$ is the same as the number of solutions to the instance $C$. The reduction of \indQueens{} to \psat{} is thus parsimonious.
	\end{corollary}
	\begin{proof}
		Each variable gadget has 2 maximum queen placements. However, once they are connected to all the clause gadgets to form a connected polyomino $P(C)$ by the process of Proposition~\ref{prop:InstanceSizePoly}, only the placements related to solutions of $C$, if there are any, are allowable. Indeed, we can read off the value of the variables from $P(C)$ by looking at the base gadgets to see if queens are at positions \textsf T or \textsf F, and similarly, we can construct a maximum queen placement from a solution to the instance $C$ by placing the first queens on the appropriate tiles of all the base variable gadgets and propagating their values.
	\end{proof}

	\section{The counting problem and rooks}\label{sec:counting}
	
	Now we will discuss the problem of counting all solutions to the \indQueens{} problem.
	
	\begin{theorem}[{\cite[Theorem~3.2]{p3sat-count-p}}]\label{thm:psatCountP}
		Counting the number of solutions of \psat{} is \#P-complete.
	\end{theorem}
	
	This implies the first part of Main Result~\ref{mainres:counting} by the properties of the reduction we exhibited.
	\begin{theorem}\label{thm:queenCountP}
		Counting the number of maximum solutions of problem \indQueens{} is \#P-complete.
	\end{theorem}
	
	\begin{proof}
		The reduction described in Section~\ref{sec:reduction} is a \emph{parsimonious reduction} by Corollary~\ref{coro:parsimonious}, which means that the number of solutions is preserved. This is also noticeable in Figure~\ref{fig:extranspoly}, where the three solutions to the \psat{} instance map to three solutions for the queens. Since by Theorem~\ref{thm:psatCountP}, \psat{} is \#P-complete, this means that \indQueens{} is also \#P-complete.
	\end{proof}
	
	And now we also have a look at the related problem of independent rook placement on a polyomino. Dropping diagonal movements for rooks makes it a much easier problem than \indQueens, in 2D at least.
	
	\begin{theorem}[{\cite[Theorem 12]{AR21}}]\label{thm:rooksP}
		The problem \indRooks{} on polyominoes can be solved in polynomial time.
	\end{theorem}
	
	We now turn into the complexity class of counting the maximum solution of \indRooks.	We will use a useful graph-theoretical result similar to the visibility representation we used in the previous section. We say that a bipartite graph $G = (X, Y; E)$ has a \textit{grid representation} if $X$ and $Y$ correspond to sets of horizontal and vertical segments in the plane such that $(x_i, y_i)\in E$ if and only if 	the segments $x_i$ and $y_i$, intersect.
	
	\begin{lemma}[{\cite[Theorem~2]{HNZ91}}]\label{lem:planar_gridintersection}
		Any planar bipartite graph admits a grid representation.
	\end{lemma}

	The idea is to convert \indRooks{} into a maximum bipartite graph matching problem, meaning we look at the vertical and horizontal lines formed by tiles of the polyomino and create an edge for each intersection. This is exactly the definition of grid intersection graphs. And since the polyomino is connected, the graphs corresponding to the polyominoes are exactly the connected grid intersection graphs, answering the first part of~\cite[Open Question 2]{AR21}.

	\begin{proposition}[{See also~\cite[Theorem~4.1]{MK21}}]\label{prop:countingmaxmatch}
		Counting the maximum matchings in planar bipartite graphs is \#P-complete.
	\end{proposition}
	\begin{proof}		
		Let $G=(U,V)$ be a planar bipartite graph. Counting all its matchings is \#P-complete \cite[Theorem 4.1]{Vadhan01}. Without loss of generality, we assume that $|U| \leq |V|$ and construct a new graph $G'$ by adding one dummy edge to each vertex in $U$, and thus also creating a new vertex in $V$ each time.
		
		Now we transform a matching $M$ in $G$ into a maximum matching $M'$ in $G'$ by including the dummy edge for all vertices in $U$ that are unmatched. For any maximum matching in $M'$, we can also invert this process by simply removing all the dummy edges and thus obtaining a matching in $G$. Since this mapping is left- and right-invertible, we have a bijection between matchings in $G$ and maximum matchings in $G'$, and the number of matchings in $G$ is the same as the number of matchings in $G'$, which has at most twice as many vertices, so its size is polynomially bounded. 
		
		We have exhibited a polynomial reduction to the problem of counting all matchings of a planar bipartite graph and thus counting the maximum matchings in a planar bipartite graph is \#P-complete as well.
	\end{proof}
	
	\begin{remark}
		This proof closely follows the one of Miklós--Krész~\cite[Theorem~4.1]{MK21}. We include it there, as the authors withdrew their preprint. Indeed, their main result, proving that counting all matchings in planar bipartite graphs is \#P-complete, was already discovered earlier by Salil P. Vadhan~\cite{Vadhan01}. However, this does not affect the validity of their proof that counting maximum matchings in planar bipartite graphs is \#P-complete. We are not aware of another reference for this specific fact, so we decided to cite it despite its withdrawal.
	\end{remark}

	We are now ready to finish the proof of Main Result~\ref{mainres:counting}. One of the most famous results of counting complexity is the fact that, while finding one perfect matching in a bipartite graph can be done in polynomial time, counting them all is \#P-complete, and so as hard as counting the solutions of a SAT formula. The same is true for \indRooks{}.
	\begin{theorem}\label{thm:rookCountP}
		Counting the number of maximum solutions of problem \indRooks{} is \#P-complete.
	\end{theorem}
	\begin{proof}
		Counting the maximum matchings in a planar bipartite graph is \#P-complete by Proposition~\ref{prop:countingmaxmatch}. Planar bipartite graphs are a subset of grid intersection graphs by Lemma~\ref{lem:planar_gridintersection}. This immediately implies that counting maximum matchings in grid intersection graphs is \#P-complete. Since counting the maximum matchings in disconnected graphs can be simply done by counting them in all (polynomially many) connected components and multiplying the result, this means that it is also \#P-complete for connected grid intersection graphs and we have completed the proof.
	\end{proof}
	
	Thus, while finding one maximum independent set is much easier for rooks than queens (unless P $=$ NP), counting all solutions is equally hard for both problems.

	\addcontentsline{toc}{section}{Acknowledgements}
	
	\section*{Acknowledgments}
	
	The authors would like to thank Érika Roldán for their support and for their collaboration on our initial work~\cite{LRMR25} on the independent queen set on polycubes, which strongly inspired this paper.
	
	ALR's research is funded by a postdoctoral research scholarship of the Fonds de Recherche du Québec -- Nature et Technologie [grant number 326641] and was funded by the Deutsche Forschungsgemeinschaft (DFG, German Research Foundation) under Germany's Excellence Strategy -- GZ 2047/2, Projekt-ID 390685813.
	
	\printbibliography
	
	\appendix
	\section{Examples of the translation}\label{sec:ex}
	
	We begin by giving, in Figure~\ref{table:3clausequeenplacement}, the 7 solutions of the 3-clause $x_a\lor x_b \lor x_c$.	
	
	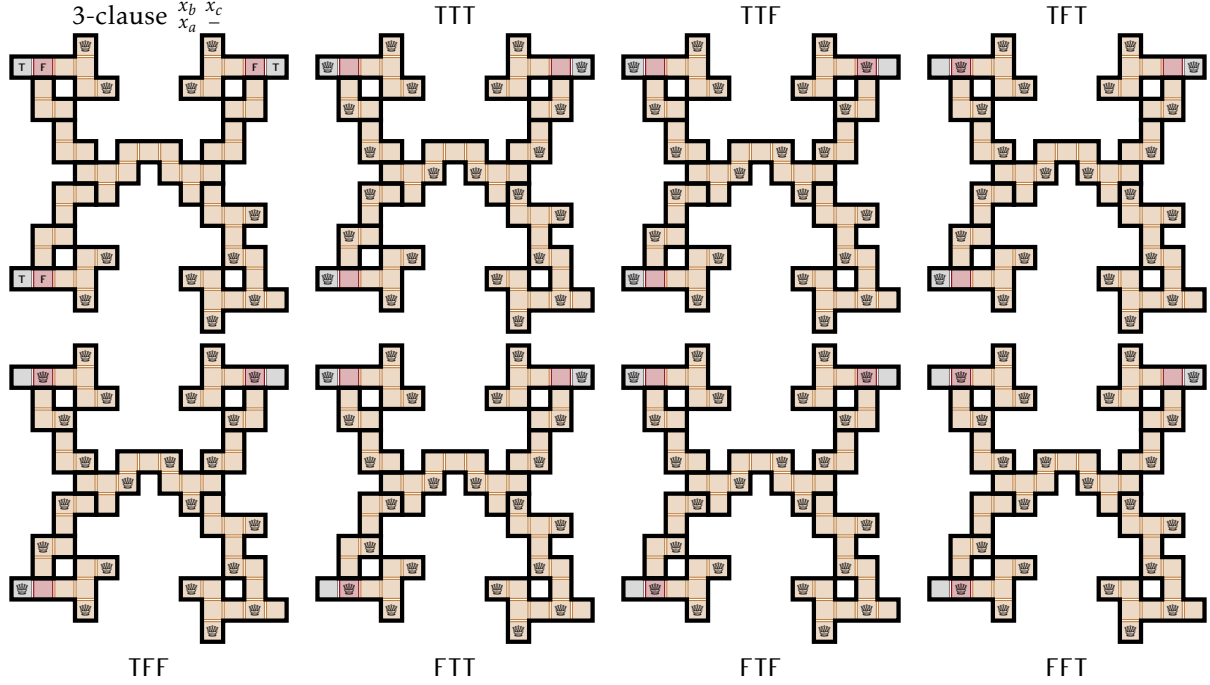
\begin{figure}[h]
		\centering
		\begin{tabular}{cccc}
			3-clause $\begin{smallmatrix}
			x_b&x_c\\x_a&-
			\end{smallmatrix}$ & \texttt{TTT} & \texttt{TTF} & \texttt{TFT}\\
			\begin{tikzpicture}[scale=.28,transform shape]
			\foreach \x/\y in {
				4/14, 10/14,
				3/13, 4/13, 10/13, 11/13, 
				2/12, 4/12, 5/12, 9/12, 10/12, 12/12,
				2/11, 3/11, 11/11, 12/11,
				3/10, 11/10,
				3/9, 4/9, 6/9, 7/9, 8/9, 10/9, 11/9,
				4/8, 5/8, 6/8, 8/8, 9/8, 10/8,
				3/7, 4/7, 5/7, 9/7, 10/7,
				3/6, 10/6, 11/6, 12/6,
				2/5, 3/5, 11/5,
				2/4, 4/4, 5/4, 11/4, 12/4,
				3/3, 4/3, 9/3, 10/3, 12/3,
				4/2, 10/2, 11/2, 12/2, 13/2,
				10/1
			} {
				\path [draw=brown!, fill=brown!30] (.5+\x-0.45, .5+\y-1.45)
				-- ++(0,.9)
				-- ++(.9,0)
				-- ++(0,-.9)
				--cycle;
			}
			
			\foreach \x/\y in {
				1/3,1/13,13/13
			} {
				\path [draw=gray!, fill=gray!30] (.5+\x-0.45, .5+\y-1.45)
				-- ++(0,.9)
				-- ++(.9,0)
				-- ++(0,-.9)
				--cycle;
				\node[anchor=center, scale=1.5] at (.5+\x, \y-.5) {\textbf{\textsf{T}}};
			}
			
			\foreach \x/\y in {
				2/3,2/13,12/13
			} {
				\path [draw=sangria!, fill=sangria!30] (.5+\x-0.45, .5+\y-1.45)
				-- ++(0,.9)
				-- ++(.9,0)
				-- ++(0,-.9)
				--cycle;
				\node[anchor=center, scale=1.5] at (.5+\x, \y-.5) {\textbf{\textsf{F}}};
			}
			
			\foreach \x/\y in {4/14, 10/14, 5/12, 9/12, 12/6, 5/4, 11/4, 9/3, 4/2, 12/2, 10/1} {
				\node[anchor=center, scale=1.8] at (.5+\x, \y-.5) {\queen};
			}
			
			\draw[line width=1.6pt, black] (1,2) -- (4,2) -- (4,1) -- (5,1) -- (5,3) -- (6,3) -- (6,4) -- (4,4) -- (4,3) -- (1,3) -- cycle; 
			\draw[line width=1.6pt, black] (1,12) -- (4,12) -- (4,11) -- (6,11) -- (6,12) -- (5,12) -- (5,14) -- (4,14) -- (4,13) -- (1,13) -- cycle; 
			
			\draw[line width=1.6pt, black] (10,0) -- (11,0) -- (11,1) -- (14,1) -- (14,2) -- (13,2) -- (13,4) -- (12,4) -- (12,5) -- (13,5) -- (13,6) -- (11,6) -- (11,7) -- (10,7) -- (10,5) -- (11,5) -- (11,3) -- (12,3) -- (12,2) -- (11,2) -- (11,3) -- (9,3) -- (9,2) -- (10,2) -- (10,0) -- cycle;
			
			\draw[line width=1.6pt, black] (14,12) -- (11,12) -- (11,11) -- (9,11) -- (9,12) -- (10,12) -- (10,14) -- (11,14) -- (11,13) -- (14,13) -- cycle; 
			\draw[line width=1.6pt, black] (2,3) -- (3,3) -- (3,4) -- (4,4) -- (4,5) -- (2,5) -- cycle; 
			\draw[line width=1.6pt, black] (3,5) -- (4,5) -- (4,6) -- (5,6) -- (5,7) -- (3,7) -- cycle; 
			\draw[line width=1.6pt, black] (2,10) -- (4,10) -- (4,11) -- (3,11) -- (3,12) -- (2,12) -- cycle; 
			\draw[line width=1.6pt, black] (3,8) -- (5,8) -- (5,9) -- (4,9) -- (4,10) -- (3,10) -- cycle; 
			\draw[line width=1.6pt, black] (13,10) -- (11,10) -- (11,11) -- (12,11) -- (12,12) -- (13,12) -- cycle; 
			\draw[line width=1.6pt, black] (12,8) -- (10,8) -- (10,9) -- (11,9) -- (11,10) -- (12,10) -- cycle; 
			
			\draw[line width=1.6pt, black] (4,7) -- (5,7) -- (5,6) -- (6,6) -- (6,7) -- (7,7) -- (7,8) -- (8,8) -- (8,7) -- (9,7) -- (9,6) -- (10,6) -- (10,7) -- (11,7) -- (11,8) -- (10,8) -- (9,8) -- (9,9) -- (6,9) -- (6,8) -- (4,8) -- cycle;
			\end{tikzpicture}
			&
			\begin{tikzpicture}[scale=.28,transform shape]
			\foreach \x/\y in {
				4/14, 10/14,
				3/13, 4/13, 10/13, 11/13, 
				2/12, 4/12, 5/12, 9/12, 10/12, 12/12,
				2/11, 3/11, 11/11, 12/11,
				3/10, 11/10,
				3/9, 4/9, 6/9, 7/9, 8/9, 10/9, 11/9,
				4/8, 5/8, 6/8, 8/8, 9/8, 10/8,
				3/7, 4/7, 5/7, 9/7, 10/7,
				3/6, 10/6, 11/6, 12/6,
				2/5, 3/5, 11/5,
				2/4, 4/4, 5/4, 11/4, 12/4,
				3/3, 4/3, 9/3, 10/3, 12/3,
				4/2, 10/2, 11/2, 12/2, 13/2,
				10/1
			} {
				\path [draw=brown!, fill=brown!30] (.5+\x-0.45, .5+\y-1.45)
				-- ++(0,.9)
				-- ++(.9,0)
				-- ++(0,-.9)
				--cycle;
			}
			
			\foreach \x/\y in {
				1/3,1/13,13/13
			} {
				\path [draw=gray!, fill=gray!30] (.5+\x-0.45, .5+\y-1.45)
				-- ++(0,.9)
				-- ++(.9,0)
				-- ++(0,-.9)
				--cycle;
			}
			
			\foreach \x/\y in {
				2/3,2/13,12/13
			} {
				\path [draw=sangria!, fill=sangria!30] (.5+\x-0.45, .5+\y-1.45)
				-- ++(0,.9)
				-- ++(.9,0)
				-- ++(0,-.9)
				--cycle;
			}
			
			\foreach \x/\y in {4/14, 10/14, 5/12, 9/12, 12/6, 5/4, 11/4, 9/3, 4/2, 12/2, 10/1,
				1/3,2/5,3/7,3/9,2/11,1/13,
				6/8,8/8,10/7,11/9,12/11,13/13
			} {
				\node[anchor=center, scale=1.8] at (.5+\x, \y-.5) {\queen};
			}
			
			\draw[line width=1.6pt, black] (1,2) -- (4,2) -- (4,1) -- (5,1) -- (5,3) -- (6,3) -- (6,4) -- (4,4) -- (4,3) -- (1,3) -- cycle; 
			\draw[line width=1.6pt, black] (1,12) -- (4,12) -- (4,11) -- (6,11) -- (6,12) -- (5,12) -- (5,14) -- (4,14) -- (4,13) -- (1,13) -- cycle; 
			
			\draw[line width=1.6pt, black] (10,0) -- (11,0) -- (11,1) -- (14,1) -- (14,2) -- (13,2) -- (13,4) -- (12,4) -- (12,5) -- (13,5) -- (13,6) -- (11,6) -- (11,7) -- (10,7) -- (10,5) -- (11,5) -- (11,3) -- (12,3) -- (12,2) -- (11,2) -- (11,3) -- (9,3) -- (9,2) -- (10,2) -- (10,0) -- cycle;
			
			\draw[line width=1.6pt, black] (14,12) -- (11,12) -- (11,11) -- (9,11) -- (9,12) -- (10,12) -- (10,14) -- (11,14) -- (11,13) -- (14,13) -- cycle; 
			\draw[line width=1.6pt, black] (2,3) -- (3,3) -- (3,4) -- (4,4) -- (4,5) -- (2,5) -- cycle; 
			\draw[line width=1.6pt, black] (3,5) -- (4,5) -- (4,6) -- (5,6) -- (5,7) -- (3,7) -- cycle; 
			\draw[line width=1.6pt, black] (2,10) -- (4,10) -- (4,11) -- (3,11) -- (3,12) -- (2,12) -- cycle; 
			\draw[line width=1.6pt, black] (3,8) -- (5,8) -- (5,9) -- (4,9) -- (4,10) -- (3,10) -- cycle; 
			\draw[line width=1.6pt, black] (13,10) -- (11,10) -- (11,11) -- (12,11) -- (12,12) -- (13,12) -- cycle; 
			\draw[line width=1.6pt, black] (12,8) -- (10,8) -- (10,9) -- (11,9) -- (11,10) -- (12,10) -- cycle; 
			
			\draw[line width=1.6pt, black] (4,7) -- (5,7) -- (5,6) -- (6,6) -- (6,7) -- (7,7) -- (7,8) -- (8,8) -- (8,7) -- (9,7) -- (9,6) -- (10,6) -- (10,7) -- (11,7) -- (11,8) -- (10,8) -- (9,8) -- (9,9) -- (6,9) -- (6,8) -- (4,8) -- cycle;
			\end{tikzpicture}
			&
			\begin{tikzpicture}[scale=.28,transform shape]
			\foreach \x/\y in {
				4/14, 10/14,
				3/13, 4/13, 10/13, 11/13, 
				2/12, 4/12, 5/12, 9/12, 10/12, 12/12,
				2/11, 3/11, 11/11, 12/11,
				3/10, 11/10,
				3/9, 4/9, 6/9, 7/9, 8/9, 10/9, 11/9,
				4/8, 5/8, 6/8, 8/8, 9/8, 10/8,
				3/7, 4/7, 5/7, 9/7, 10/7,
				3/6, 10/6, 11/6, 12/6,
				2/5, 3/5, 11/5,
				2/4, 4/4, 5/4, 11/4, 12/4,
				3/3, 4/3, 9/3, 10/3, 12/3,
				4/2, 10/2, 11/2, 12/2, 13/2,
				10/1
			} {
				\path [draw=brown!, fill=brown!30] (.5+\x-0.45, .5+\y-1.45)
				-- ++(0,.9)
				-- ++(.9,0)
				-- ++(0,-.9)
				--cycle;
			}
			
			\foreach \x/\y in {
				1/3,1/13,13/13
			} {
				\path [draw=gray!, fill=gray!30] (.5+\x-0.45, .5+\y-1.45)
				-- ++(0,.9)
				-- ++(.9,0)
				-- ++(0,-.9)
				--cycle;
			}
			
			\foreach \x/\y in {
				2/3,2/13,12/13
			} {
				\path [draw=sangria!, fill=sangria!30] (.5+\x-0.45, .5+\y-1.45)
				-- ++(0,.9)
				-- ++(.9,0)
				-- ++(0,-.9)
				--cycle;
			}
			
			\foreach \x/\y in {4/14, 10/14, 5/12, 9/12, 12/6, 5/4, 11/4, 9/3, 4/2, 12/2, 10/1,
				1/3,2/5,3/7,3/9,2/11,1/13,
				6/8,8/9,9/7,10/9,11/11,12/13
			} {
				\node[anchor=center, scale=1.8] at (.5+\x, \y-.5) {\queen};
			}
			
			\draw[line width=1.6pt, black] (1,2) -- (4,2) -- (4,1) -- (5,1) -- (5,3) -- (6,3) -- (6,4) -- (4,4) -- (4,3) -- (1,3) -- cycle; 
			\draw[line width=1.6pt, black] (1,12) -- (4,12) -- (4,11) -- (6,11) -- (6,12) -- (5,12) -- (5,14) -- (4,14) -- (4,13) -- (1,13) -- cycle; 
			
			\draw[line width=1.6pt, black] (10,0) -- (11,0) -- (11,1) -- (14,1) -- (14,2) -- (13,2) -- (13,4) -- (12,4) -- (12,5) -- (13,5) -- (13,6) -- (11,6) -- (11,7) -- (10,7) -- (10,5) -- (11,5) -- (11,3) -- (12,3) -- (12,2) -- (11,2) -- (11,3) -- (9,3) -- (9,2) -- (10,2) -- (10,0) -- cycle;
			
			\draw[line width=1.6pt, black] (14,12) -- (11,12) -- (11,11) -- (9,11) -- (9,12) -- (10,12) -- (10,14) -- (11,14) -- (11,13) -- (14,13) -- cycle; 
			\draw[line width=1.6pt, black] (2,3) -- (3,3) -- (3,4) -- (4,4) -- (4,5) -- (2,5) -- cycle; 
			\draw[line width=1.6pt, black] (3,5) -- (4,5) -- (4,6) -- (5,6) -- (5,7) -- (3,7) -- cycle; 
			\draw[line width=1.6pt, black] (2,10) -- (4,10) -- (4,11) -- (3,11) -- (3,12) -- (2,12) -- cycle; 
			\draw[line width=1.6pt, black] (3,8) -- (5,8) -- (5,9) -- (4,9) -- (4,10) -- (3,10) -- cycle; 
			\draw[line width=1.6pt, black] (13,10) -- (11,10) -- (11,11) -- (12,11) -- (12,12) -- (13,12) -- cycle; 
			\draw[line width=1.6pt, black] (12,8) -- (10,8) -- (10,9) -- (11,9) -- (11,10) -- (12,10) -- cycle; 
			
			\draw[line width=1.6pt, black] (4,7) -- (5,7) -- (5,6) -- (6,6) -- (6,7) -- (7,7) -- (7,8) -- (8,8) -- (8,7) -- (9,7) -- (9,6) -- (10,6) -- (10,7) -- (11,7) -- (11,8) -- (10,8) -- (9,8) -- (9,9) -- (6,9) -- (6,8) -- (4,8) -- cycle;
			\end{tikzpicture}
			&
			\begin{tikzpicture}[scale=.28,transform shape]
			\foreach \x/\y in {
				4/14, 10/14,
				3/13, 4/13, 10/13, 11/13, 
				2/12, 4/12, 5/12, 9/12, 10/12, 12/12,
				2/11, 3/11, 11/11, 12/11,
				3/10, 11/10,
				3/9, 4/9, 6/9, 7/9, 8/9, 10/9, 11/9,
				4/8, 5/8, 6/8, 8/8, 9/8, 10/8,
				3/7, 4/7, 5/7, 9/7, 10/7,
				3/6, 10/6, 11/6, 12/6,
				2/5, 3/5, 11/5,
				2/4, 4/4, 5/4, 11/4, 12/4,
				3/3, 4/3, 9/3, 10/3, 12/3,
				4/2, 10/2, 11/2, 12/2, 13/2,
				10/1
			} {
				\path [draw=brown!, fill=brown!30] (.5+\x-0.45, .5+\y-1.45)
				-- ++(0,.9)
				-- ++(.9,0)
				-- ++(0,-.9)
				--cycle;
			}
			
			\foreach \x/\y in {
				1/3,1/13,13/13
			} {
				\path [draw=gray!, fill=gray!30] (.5+\x-0.45, .5+\y-1.45)
				-- ++(0,.9)
				-- ++(.9,0)
				-- ++(0,-.9)
				--cycle;
			}
			
			\foreach \x/\y in {
				2/3,2/13,12/13
			} {
				\path [draw=sangria!, fill=sangria!30] (.5+\x-0.45, .5+\y-1.45)
				-- ++(0,.9)
				-- ++(.9,0)
				-- ++(0,-.9)
				--cycle;
			}
			
			\foreach \x/\y in {4/14, 10/14, 5/12, 9/12, 12/6, 5/4, 11/4, 9/3, 4/2, 12/2, 10/1,
				1/3,2/5,3/7,4/9,3/11,2/13,
				6/8,8/8,10/7,11/9,12/11,13/13
			} {
				\node[anchor=center, scale=1.8] at (.5+\x, \y-.5) {\queen};
			}
			
			\draw[line width=1.6pt, black] (1,2) -- (4,2) -- (4,1) -- (5,1) -- (5,3) -- (6,3) -- (6,4) -- (4,4) -- (4,3) -- (1,3) -- cycle; 
			\draw[line width=1.6pt, black] (1,12) -- (4,12) -- (4,11) -- (6,11) -- (6,12) -- (5,12) -- (5,14) -- (4,14) -- (4,13) -- (1,13) -- cycle; 
			
			\draw[line width=1.6pt, black] (10,0) -- (11,0) -- (11,1) -- (14,1) -- (14,2) -- (13,2) -- (13,4) -- (12,4) -- (12,5) -- (13,5) -- (13,6) -- (11,6) -- (11,7) -- (10,7) -- (10,5) -- (11,5) -- (11,3) -- (12,3) -- (12,2) -- (11,2) -- (11,3) -- (9,3) -- (9,2) -- (10,2) -- (10,0) -- cycle;
			
			\draw[line width=1.6pt, black] (14,12) -- (11,12) -- (11,11) -- (9,11) -- (9,12) -- (10,12) -- (10,14) -- (11,14) -- (11,13) -- (14,13) -- cycle; 
			\draw[line width=1.6pt, black] (2,3) -- (3,3) -- (3,4) -- (4,4) -- (4,5) -- (2,5) -- cycle; 
			\draw[line width=1.6pt, black] (3,5) -- (4,5) -- (4,6) -- (5,6) -- (5,7) -- (3,7) -- cycle; 
			\draw[line width=1.6pt, black] (2,10) -- (4,10) -- (4,11) -- (3,11) -- (3,12) -- (2,12) -- cycle; 
			\draw[line width=1.6pt, black] (3,8) -- (5,8) -- (5,9) -- (4,9) -- (4,10) -- (3,10) -- cycle; 
			\draw[line width=1.6pt, black] (13,10) -- (11,10) -- (11,11) -- (12,11) -- (12,12) -- (13,12) -- cycle; 
			\draw[line width=1.6pt, black] (12,8) -- (10,8) -- (10,9) -- (11,9) -- (11,10) -- (12,10) -- cycle; 
			
			\draw[line width=1.6pt, black] (4,7) -- (5,7) -- (5,6) -- (6,6) -- (6,7) -- (7,7) -- (7,8) -- (8,8) -- (8,7) -- (9,7) -- (9,6) -- (10,6) -- (10,7) -- (11,7) -- (11,8) -- (10,8) -- (9,8) -- (9,9) -- (6,9) -- (6,8) -- (4,8) -- cycle;
			\end{tikzpicture}
			\\
			\begin{tikzpicture}[scale=.28,transform shape]
			\foreach \x/\y in {
				4/14, 10/14,
				3/13, 4/13, 10/13, 11/13, 
				2/12, 4/12, 5/12, 9/12, 10/12, 12/12,
				2/11, 3/11, 11/11, 12/11,
				3/10, 11/10,
				3/9, 4/9, 6/9, 7/9, 8/9, 10/9, 11/9,
				4/8, 5/8, 6/8, 8/8, 9/8, 10/8,
				3/7, 4/7, 5/7, 9/7, 10/7,
				3/6, 10/6, 11/6, 12/6,
				2/5, 3/5, 11/5,
				2/4, 4/4, 5/4, 11/4, 12/4,
				3/3, 4/3, 9/3, 10/3, 12/3,
				4/2, 10/2, 11/2, 12/2, 13/2,
				10/1
			} {
				\path [draw=brown!, fill=brown!30] (.5+\x-0.45, .5+\y-1.45)
				-- ++(0,.9)
				-- ++(.9,0)
				-- ++(0,-.9)
				--cycle;
			}
			
			\foreach \x/\y in {
				1/3,1/13,13/13
			} {
				\path [draw=gray!, fill=gray!30] (.5+\x-0.45, .5+\y-1.45)
				-- ++(0,.9)
				-- ++(.9,0)
				-- ++(0,-.9)
				--cycle;
			}
			
			\foreach \x/\y in {
				2/3,2/13,12/13
			} {
				\path [draw=sangria!, fill=sangria!30] (.5+\x-0.45, .5+\y-1.45)
				-- ++(0,.9)
				-- ++(.9,0)
				-- ++(0,-.9)
				--cycle;
			}
			
			\foreach \x/\y in {4/14, 10/14, 5/12, 9/12, 12/6, 5/4, 11/4, 9/3, 4/2, 12/2, 10/1,
				1/3,2/5,3/7,4/9,3/11,2/13,
				6/8,8/9,9/7,10/9,11/11,12/13
			} {
				\node[anchor=center, scale=1.8] at (.5+\x, \y-.5) {\queen};
			}
			
			\draw[line width=1.6pt, black] (1,2) -- (4,2) -- (4,1) -- (5,1) -- (5,3) -- (6,3) -- (6,4) -- (4,4) -- (4,3) -- (1,3) -- cycle; 
			\draw[line width=1.6pt, black] (1,12) -- (4,12) -- (4,11) -- (6,11) -- (6,12) -- (5,12) -- (5,14) -- (4,14) -- (4,13) -- (1,13) -- cycle; 
			
			\draw[line width=1.6pt, black] (10,0) -- (11,0) -- (11,1) -- (14,1) -- (14,2) -- (13,2) -- (13,4) -- (12,4) -- (12,5) -- (13,5) -- (13,6) -- (11,6) -- (11,7) -- (10,7) -- (10,5) -- (11,5) -- (11,3) -- (12,3) -- (12,2) -- (11,2) -- (11,3) -- (9,3) -- (9,2) -- (10,2) -- (10,0) -- cycle;
			
			\draw[line width=1.6pt, black] (14,12) -- (11,12) -- (11,11) -- (9,11) -- (9,12) -- (10,12) -- (10,14) -- (11,14) -- (11,13) -- (14,13) -- cycle; 
			\draw[line width=1.6pt, black] (2,3) -- (3,3) -- (3,4) -- (4,4) -- (4,5) -- (2,5) -- cycle; 
			\draw[line width=1.6pt, black] (3,5) -- (4,5) -- (4,6) -- (5,6) -- (5,7) -- (3,7) -- cycle; 
			\draw[line width=1.6pt, black] (2,10) -- (4,10) -- (4,11) -- (3,11) -- (3,12) -- (2,12) -- cycle; 
			\draw[line width=1.6pt, black] (3,8) -- (5,8) -- (5,9) -- (4,9) -- (4,10) -- (3,10) -- cycle; 
			\draw[line width=1.6pt, black] (13,10) -- (11,10) -- (11,11) -- (12,11) -- (12,12) -- (13,12) -- cycle; 
			\draw[line width=1.6pt, black] (12,8) -- (10,8) -- (10,9) -- (11,9) -- (11,10) -- (12,10) -- cycle; 
			
			\draw[line width=1.6pt, black] (4,7) -- (5,7) -- (5,6) -- (6,6) -- (6,7) -- (7,7) -- (7,8) -- (8,8) -- (8,7) -- (9,7) -- (9,6) -- (10,6) -- (10,7) -- (11,7) -- (11,8) -- (10,8) -- (9,8) -- (9,9) -- (6,9) -- (6,8) -- (4,8) -- cycle;
			\end{tikzpicture}
			&
			\begin{tikzpicture}[scale=.28,transform shape]
			\foreach \x/\y in {
				4/14, 10/14,
				3/13, 4/13, 10/13, 11/13, 
				2/12, 4/12, 5/12, 9/12, 10/12, 12/12,
				2/11, 3/11, 11/11, 12/11,
				3/10, 11/10,
				3/9, 4/9, 6/9, 7/9, 8/9, 10/9, 11/9,
				4/8, 5/8, 6/8, 8/8, 9/8, 10/8,
				3/7, 4/7, 5/7, 9/7, 10/7,
				3/6, 10/6, 11/6, 12/6,
				2/5, 3/5, 11/5,
				2/4, 4/4, 5/4, 11/4, 12/4,
				3/3, 4/3, 9/3, 10/3, 12/3,
				4/2, 10/2, 11/2, 12/2, 13/2,
				10/1
			} {
				\path [draw=brown!, fill=brown!30] (.5+\x-0.45, .5+\y-1.45)
				-- ++(0,.9)
				-- ++(.9,0)
				-- ++(0,-.9)
				--cycle;
			}
			
			\foreach \x/\y in {
				1/3,1/13,13/13
			} {
				\path [draw=gray!, fill=gray!30] (.5+\x-0.45, .5+\y-1.45)
				-- ++(0,.9)
				-- ++(.9,0)
				-- ++(0,-.9)
				--cycle;
			}
			
			\foreach \x/\y in {
				2/3,2/13,12/13
			} {
				\path [draw=sangria!, fill=sangria!30] (.5+\x-0.45, .5+\y-1.45)
				-- ++(0,.9)
				-- ++(.9,0)
				-- ++(0,-.9)
				--cycle;
			}
			
			\foreach \x/\y in {4/14, 10/14, 5/12, 9/12, 12/6, 5/4, 11/4, 9/3, 4/2, 12/2, 10/1,
				2/3,3/5,4/7,3/9,2/11,1/13,
				6/8,8/8,10/7,11/9,12/11,13/13
			} {
				\node[anchor=center, scale=1.8] at (.5+\x, \y-.5) {\queen};
			}
			
			\draw[line width=1.6pt, black] (1,2) -- (4,2) -- (4,1) -- (5,1) -- (5,3) -- (6,3) -- (6,4) -- (4,4) -- (4,3) -- (1,3) -- cycle; 
			\draw[line width=1.6pt, black] (1,12) -- (4,12) -- (4,11) -- (6,11) -- (6,12) -- (5,12) -- (5,14) -- (4,14) -- (4,13) -- (1,13) -- cycle; 
			
			\draw[line width=1.6pt, black] (10,0) -- (11,0) -- (11,1) -- (14,1) -- (14,2) -- (13,2) -- (13,4) -- (12,4) -- (12,5) -- (13,5) -- (13,6) -- (11,6) -- (11,7) -- (10,7) -- (10,5) -- (11,5) -- (11,3) -- (12,3) -- (12,2) -- (11,2) -- (11,3) -- (9,3) -- (9,2) -- (10,2) -- (10,0) -- cycle;
			
			\draw[line width=1.6pt, black] (14,12) -- (11,12) -- (11,11) -- (9,11) -- (9,12) -- (10,12) -- (10,14) -- (11,14) -- (11,13) -- (14,13) -- cycle; 
			\draw[line width=1.6pt, black] (2,3) -- (3,3) -- (3,4) -- (4,4) -- (4,5) -- (2,5) -- cycle; 
			\draw[line width=1.6pt, black] (3,5) -- (4,5) -- (4,6) -- (5,6) -- (5,7) -- (3,7) -- cycle; 
			\draw[line width=1.6pt, black] (2,10) -- (4,10) -- (4,11) -- (3,11) -- (3,12) -- (2,12) -- cycle; 
			\draw[line width=1.6pt, black] (3,8) -- (5,8) -- (5,9) -- (4,9) -- (4,10) -- (3,10) -- cycle; 
			\draw[line width=1.6pt, black] (13,10) -- (11,10) -- (11,11) -- (12,11) -- (12,12) -- (13,12) -- cycle; 
			\draw[line width=1.6pt, black] (12,8) -- (10,8) -- (10,9) -- (11,9) -- (11,10) -- (12,10) -- cycle; 
			
			\draw[line width=1.6pt, black] (4,7) -- (5,7) -- (5,6) -- (6,6) -- (6,7) -- (7,7) -- (7,8) -- (8,8) -- (8,7) -- (9,7) -- (9,6) -- (10,6) -- (10,7) -- (11,7) -- (11,8) -- (10,8) -- (9,8) -- (9,9) -- (6,9) -- (6,8) -- (4,8) -- cycle;
			\end{tikzpicture}
			&
			\begin{tikzpicture}[scale=.28,transform shape]
			\foreach \x/\y in {
				4/14, 10/14,
				3/13, 4/13, 10/13, 11/13, 
				2/12, 4/12, 5/12, 9/12, 10/12, 12/12,
				2/11, 3/11, 11/11, 12/11,
				3/10, 11/10,
				3/9, 4/9, 6/9, 7/9, 8/9, 10/9, 11/9,
				4/8, 5/8, 6/8, 8/8, 9/8, 10/8,
				3/7, 4/7, 5/7, 9/7, 10/7,
				3/6, 10/6, 11/6, 12/6,
				2/5, 3/5, 11/5,
				2/4, 4/4, 5/4, 11/4, 12/4,
				3/3, 4/3, 9/3, 10/3, 12/3,
				4/2, 10/2, 11/2, 12/2, 13/2,
				10/1
			} {
				\path [draw=brown!, fill=brown!30] (.5+\x-0.45, .5+\y-1.45)
				-- ++(0,.9)
				-- ++(.9,0)
				-- ++(0,-.9)
				--cycle;
			}
			
			\foreach \x/\y in {
				1/3,1/13,13/13
			} {
				\path [draw=gray!, fill=gray!30] (.5+\x-0.45, .5+\y-1.45)
				-- ++(0,.9)
				-- ++(.9,0)
				-- ++(0,-.9)
				--cycle;
			}
			
			\foreach \x/\y in {
				2/3,2/13,12/13
			} {
				\path [draw=sangria!, fill=sangria!30] (.5+\x-0.45, .5+\y-1.45)
				-- ++(0,.9)
				-- ++(.9,0)
				-- ++(0,-.9)
				--cycle;
			}
			
			\foreach \x/\y in {4/14, 10/14, 5/12, 9/12, 12/6, 5/4, 11/4, 9/3, 4/2, 12/2, 10/1,
				2/3,3/5,4/7,3/9,2/11,1/13,
				6/8,8/9,9/7,10/9,11/11,12/13
			} {
				\node[anchor=center, scale=1.8] at (.5+\x, \y-.5) {\queen};
			}
			
			\draw[line width=1.6pt, black] (1,2) -- (4,2) -- (4,1) -- (5,1) -- (5,3) -- (6,3) -- (6,4) -- (4,4) -- (4,3) -- (1,3) -- cycle; 
			\draw[line width=1.6pt, black] (1,12) -- (4,12) -- (4,11) -- (6,11) -- (6,12) -- (5,12) -- (5,14) -- (4,14) -- (4,13) -- (1,13) -- cycle; 
			
			\draw[line width=1.6pt, black] (10,0) -- (11,0) -- (11,1) -- (14,1) -- (14,2) -- (13,2) -- (13,4) -- (12,4) -- (12,5) -- (13,5) -- (13,6) -- (11,6) -- (11,7) -- (10,7) -- (10,5) -- (11,5) -- (11,3) -- (12,3) -- (12,2) -- (11,2) -- (11,3) -- (9,3) -- (9,2) -- (10,2) -- (10,0) -- cycle;
			
			\draw[line width=1.6pt, black] (14,12) -- (11,12) -- (11,11) -- (9,11) -- (9,12) -- (10,12) -- (10,14) -- (11,14) -- (11,13) -- (14,13) -- cycle; 
			\draw[line width=1.6pt, black] (2,3) -- (3,3) -- (3,4) -- (4,4) -- (4,5) -- (2,5) -- cycle; 
			\draw[line width=1.6pt, black] (3,5) -- (4,5) -- (4,6) -- (5,6) -- (5,7) -- (3,7) -- cycle; 
			\draw[line width=1.6pt, black] (2,10) -- (4,10) -- (4,11) -- (3,11) -- (3,12) -- (2,12) -- cycle; 
			\draw[line width=1.6pt, black] (3,8) -- (5,8) -- (5,9) -- (4,9) -- (4,10) -- (3,10) -- cycle; 
			\draw[line width=1.6pt, black] (13,10) -- (11,10) -- (11,11) -- (12,11) -- (12,12) -- (13,12) -- cycle; 
			\draw[line width=1.6pt, black] (12,8) -- (10,8) -- (10,9) -- (11,9) -- (11,10) -- (12,10) -- cycle; 
			
			\draw[line width=1.6pt, black] (4,7) -- (5,7) -- (5,6) -- (6,6) -- (6,7) -- (7,7) -- (7,8) -- (8,8) -- (8,7) -- (9,7) -- (9,6) -- (10,6) -- (10,7) -- (11,7) -- (11,8) -- (10,8) -- (9,8) -- (9,9) -- (6,9) -- (6,8) -- (4,8) -- cycle;
			\end{tikzpicture}
			&
			\begin{tikzpicture}[scale=.28,transform shape]
			\foreach \x/\y in {
				4/14, 10/14,
				3/13, 4/13, 10/13, 11/13, 
				2/12, 4/12, 5/12, 9/12, 10/12, 12/12,
				2/11, 3/11, 11/11, 12/11,
				3/10, 11/10,
				3/9, 4/9, 6/9, 7/9, 8/9, 10/9, 11/9,
				4/8, 5/8, 6/8, 8/8, 9/8, 10/8,
				3/7, 4/7, 5/7, 9/7, 10/7,
				3/6, 10/6, 11/6, 12/6,
				2/5, 3/5, 11/5,
				2/4, 4/4, 5/4, 11/4, 12/4,
				3/3, 4/3, 9/3, 10/3, 12/3,
				4/2, 10/2, 11/2, 12/2, 13/2,
				10/1
			} {
				\path [draw=brown!, fill=brown!30] (.5+\x-0.45, .5+\y-1.45)
				-- ++(0,.9)
				-- ++(.9,0)
				-- ++(0,-.9)
				--cycle;
			}
			
			\foreach \x/\y in {
				1/3,1/13,13/13
			} {
				\path [draw=gray!, fill=gray!30] (.5+\x-0.45, .5+\y-1.45)
				-- ++(0,.9)
				-- ++(.9,0)
				-- ++(0,-.9)
				--cycle;
			}
			
			\foreach \x/\y in {
				2/3,2/13,12/13
			} {
				\path [draw=sangria!, fill=sangria!30] (.5+\x-0.45, .5+\y-1.45)
				-- ++(0,.9)
				-- ++(.9,0)
				-- ++(0,-.9)
				--cycle;
			}
			
			\foreach \x/\y in {4/14, 10/14, 5/12, 9/12, 12/6, 5/4, 11/4, 9/3, 4/2, 12/2, 10/1,
				2/3,3/5,5/7,4/9,3/11,2/13,
				6/9,8/8,10/7,11/9,12/11,13/13
			} {
				\node[anchor=center, scale=1.8] at (.5+\x, \y-.5) {\queen};
			}
			
			\draw[line width=1.6pt, black] (1,2) -- (4,2) -- (4,1) -- (5,1) -- (5,3) -- (6,3) -- (6,4) -- (4,4) -- (4,3) -- (1,3) -- cycle; 
			\draw[line width=1.6pt, black] (1,12) -- (4,12) -- (4,11) -- (6,11) -- (6,12) -- (5,12) -- (5,14) -- (4,14) -- (4,13) -- (1,13) -- cycle; 
			
			\draw[line width=1.6pt, black] (10,0) -- (11,0) -- (11,1) -- (14,1) -- (14,2) -- (13,2) -- (13,4) -- (12,4) -- (12,5) -- (13,5) -- (13,6) -- (11,6) -- (11,7) -- (10,7) -- (10,5) -- (11,5) -- (11,3) -- (12,3) -- (12,2) -- (11,2) -- (11,3) -- (9,3) -- (9,2) -- (10,2) -- (10,0) -- cycle;
			
			\draw[line width=1.6pt, black] (14,12) -- (11,12) -- (11,11) -- (9,11) -- (9,12) -- (10,12) -- (10,14) -- (11,14) -- (11,13) -- (14,13) -- cycle; 
			\draw[line width=1.6pt, black] (2,3) -- (3,3) -- (3,4) -- (4,4) -- (4,5) -- (2,5) -- cycle; 
			\draw[line width=1.6pt, black] (3,5) -- (4,5) -- (4,6) -- (5,6) -- (5,7) -- (3,7) -- cycle; 
			\draw[line width=1.6pt, black] (2,10) -- (4,10) -- (4,11) -- (3,11) -- (3,12) -- (2,12) -- cycle; 
			\draw[line width=1.6pt, black] (3,8) -- (5,8) -- (5,9) -- (4,9) -- (4,10) -- (3,10) -- cycle; 
			\draw[line width=1.6pt, black] (13,10) -- (11,10) -- (11,11) -- (12,11) -- (12,12) -- (13,12) -- cycle; 
			\draw[line width=1.6pt, black] (12,8) -- (10,8) -- (10,9) -- (11,9) -- (11,10) -- (12,10) -- cycle; 
			
			\draw[line width=1.6pt, black] (4,7) -- (5,7) -- (5,6) -- (6,6) -- (6,7) -- (7,7) -- (7,8) -- (8,8) -- (8,7) -- (9,7) -- (9,6) -- (10,6) -- (10,7) -- (11,7) -- (11,8) -- (10,8) -- (9,8) -- (9,9) -- (6,9) -- (6,8) -- (4,8) -- cycle;
			\end{tikzpicture}
			\\
			\texttt{TFF} & \texttt{FTT} & \texttt{FTF} & \texttt{FFT}
		\end{tabular}
		\caption{The 7 maximum queen placements for the 3-clauses $x_a\lor x_b \lor x_c$(convention for reading the clause: $x_a$ is placed bottom left, $x_b$  upper left, and $x_c$ upper right).}\label{table:3clausequeenplacement}
	\end{figure}

	We now give a complete example in Figure~\ref{fig:extranspoly} of the reduction discussed in Section~\ref{sec:reduction} by exhibiting a polyomino for the simple instance $C$ of \psat{} of Figure~\ref{fig:ex_psat}. For easy comparison, we chose the same instance as the examples in~\cite{LRMR25}.

	Note that, in order to keep a reasonably sized example, the polyomino of Figure~\ref{fig:extranspoly} is \emph{not} exactly the construction we describe in Proposition~\ref{prop:InstanceSizePoly}. In particular, we only used the shift gadgets in the dotted box. This enabled us to reduce the size of the polyomino as we did not enlarge the boxes of the grid path version of instance $C$, as in Figure~\ref{fig:path_graph-psat}.
	
	Indeed, the construction imposes some boxes to enable for placing the signal at a convenient spot for joining it with another. Here we reduced the size of the bounding boxes so that it fits the page, and instead used some small \textit{ad hoc} optimisations. The example is also available on the website~\cite{LRM26Software} to see the different maximum queen placements. There are 3 maximum queen placements of 258 queens on the 798-tile polyomino. They correspond to the 3 solutions (\texttt{T,F,T}), (\texttt{T,F,F}), (\texttt{F,T,F}) of the instance $C$ of Figure~\ref{fig:ex_psat}.
	
	\newcommand{\figscale}{0.19}

	

	%

	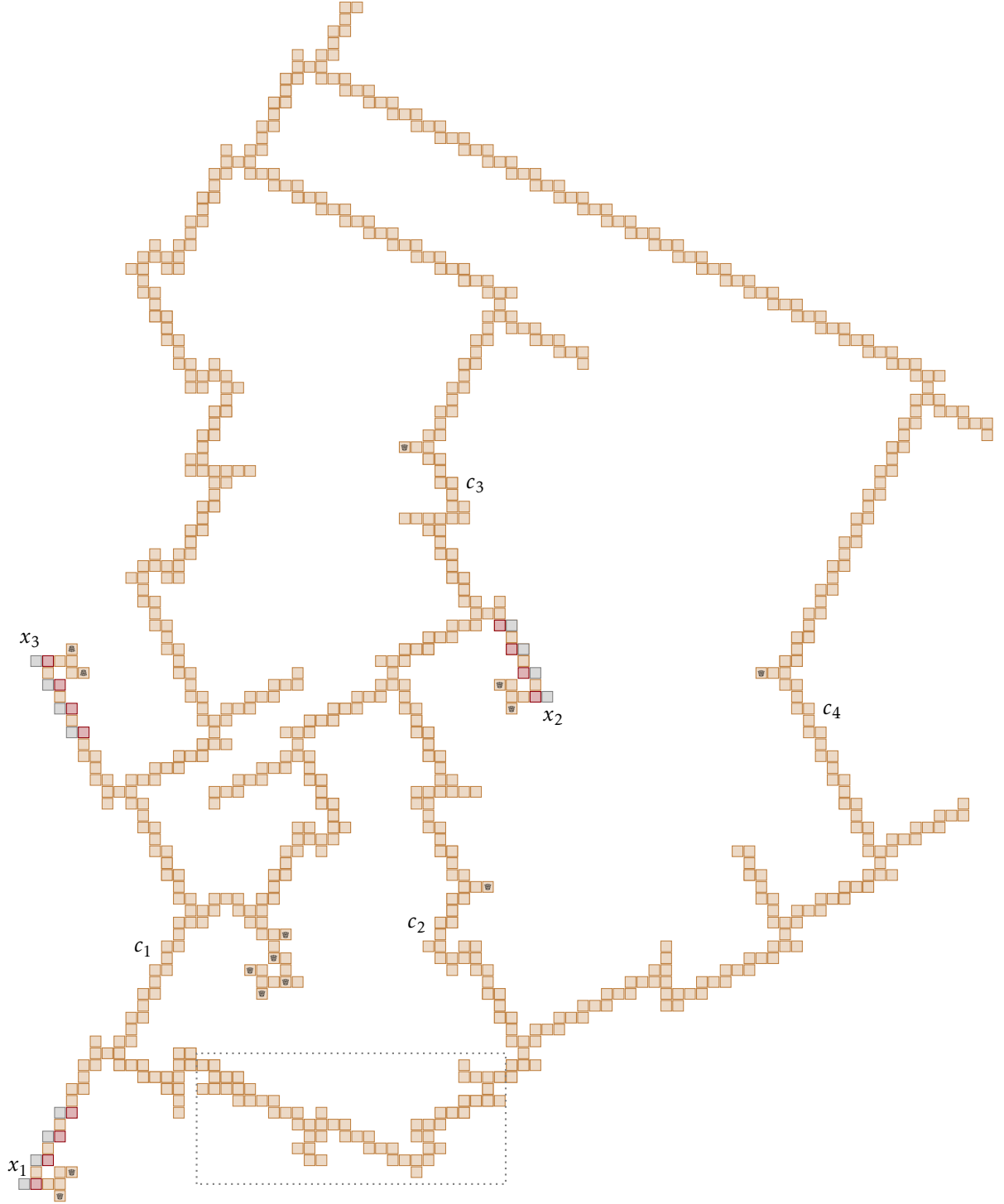
\begin{figure}[p]
		\centering
		\begin{tikzpicture}[scale=\figscale]
		\path (0,3) pic{var} node[above=3pt]{$x_1$};
		\path (1,48) pic[rotate=180,xscale=-1]{var} node[above]{$x_3$};
		\path (5,41) pic[rotate=0,yscale=-1]{splitter};
		\path (12,22) pic {threeclause} node[above left]{$c_1$};
		\path (4,10) pic {splittersmall};
		\path (11,20) pic {ext};
		\path (12,38) pic[rotate=270, xscale=-1]{splitter} ;
		\path (13,49) pic[xscale=-1] {turner};
		\foreach \x/\y in {24/33,25/35} 
		{
		}
		\path (36+10-8+5,36+8+3) pic[xscale=-1] {splitter};
		\path (41+10-8+2,30+8+6) pic[xscale=-1]{var} node[below]{$x_2$};
		\path (35-11,25+14) pic[xscale=-1,rotate = 180] {turner};
		\path (36,36+14) pic[xscale=-1,rotate=270] {splitter};
		\path (35-7,39+5) pic[rotate = 90, xscale=-1] {splitter};
		\path (14,57) pic{inverter}; 
		\path (15,85) pic{splitter};
		\foreach \x/\y in {27/84,29/83,31/82,33/81,35/80,37/79} 
		{
			\path (\x-1,\y+2) pic[rotate=270]{ext};
		}, 11/4
		\path (36,90-9) pic[rotate=270]{splitter}; 
		\path (15,65) pic{turner};
		\path (13,75) pic[xscale=-1]{turner};
		\path (38,53) pic[xscale=-1]{inverter}; 
		\path (37,61) pic[xscale=-1] {twoclause} node[above right]{$c_3$};
		\path (33,43) pic[rotate=180,xscale=-1] {inverter}; 
		\path (35,24) pic{twoclause} node[above left] {$c_2$};
		\foreach \x/\y in {
		} 
		{
			\path (\x,\y+2) pic[rotate=270]{ext};
		}
		\path (7,15) pic[rotate=270]{inverter};
		\path (14,16) pic [rotate=270] {forwardshift};
		\path (37-1+2,31-2) pic[rotate=180] {turner};
		\path (34+2+2,11) pic[rotate=270,xscale=-1] {splitterverysmall};
		\path (36,10) pic {rightshift};
		\path (28,8) pic [rotate=270]{turner};
		\path (21,93) pic {splittersmall};
		\foreach \x/\y in {30/94,32/93,34/92,36/91,38/90,40/89,42/88,44/87,46/86,48/85,
			50/84,
			52/83,54/82,56/81,58/80,60/79,
			62/78,64/77,66/76,68/75,70/74} 
		{
			\path (\x,\y+1) pic[rotate=270]{ext};
		}
		\path (72,74) pic[rotate=270] {splittersmall};
		\foreach \x/\y in {73/64,72/62,71/60,70/58,69/56,68/54,67/52,66/50,65/48
		} 
		{
			\path (\x+1,\y+1) pic[rotate=180]{ext};
		}
		\path (50-2,19-1) pic[xscale=-1,rotate=90] {inverter};
		\path (62-2,22-1) pic[xscale=-1,rotate=90] {splitter};
		\path (62+6,22+5) pic[xscale=-1,rotate=90] {splitter};
		\path (69,38) pic[xscale=-1] {ext};
		\path (68,40) pic[xscale=-1] {ext};
		\path (66+1,42) pic[xscale=-1] {twoclause} node[above right]{$c_4$};
		\draw [dotted, thick, gray] (15,3.5) rectangle (41,14.5);
		\end{tikzpicture}
		\caption{A translation of the instance of \psat{} from Figure~\ref{fig:ex_psat}. There are 3 maximum placements of 258 independent queens that guard this 798-polyomino, corresponding to the 3 solutions ($(x_1,x_2,x_3) \in \{(\texttt{T,F,T}), (\texttt{T,F,F}), (\texttt{F,T,F})\}$).
		}\label{fig:extranspoly}
	\end{figure}

\end{document}